\documentclass[11pt]{llncs}
\usepackage[final]{graphicx}
\usepackage[margin=1in]{geometry}
\usepackage{xcolor} 
\usepackage{comment} 
\usepackage{amssymb,amsfonts,amsmath,latexsym}
\usepackage{hyperref}
\hypersetup{
     colorlinks = true,
     linkcolor = blue,
     anchorcolor = blue,
     citecolor = blue,
     filecolor = blue,
     urlcolor = blue,
     }
\usepackage{algorithm,algorithmic}

\usepackage{natbib}
\let\cite\citep
\def\eps{\epsilon}
\def\C{\mathcal{C}}
\def\V{\mathcal{V}}
\def\I{\mathcal{I}}

\def\cost{\mathrm{cost}}
\def\top{\mathrm{top}}
\def\dist{\mathrm{dist}}

\def\Egal{\mathrm{EC}}
\def\Util{\mathrm{SC}}
\def\weight{\mathrm{wt}}

\def\Cl{\mathrm{Cluster}}
\def\reals{\mathbb{R}}

\def\OPT{\mathrm{OPT}}
\def\A{\mathrm{Alg}}
\def\cres{\text{cr}}

\newcommand*{\xdash}[1][3em]{\rule[0.5ex]{#1}{0.55pt}}

\newtheorem{observation}{Observation}
\begin{document}
\title{On the Distortion of Committee Election with $1$-Euclidean Preferences and Few Distance Queries%
\thanks{This work has been supported by project MIS 5154714 of the National Recovery and Resilience Plan Greece 2.0 funded by the European Union under the NextGenerationEU Program and by the Hellenic Foundation for Research and Innovation (H.F.R.I.) under the ``First Call for H.F.R.I. Research Projects to support Faculty members and Researchers and the procurement of high-cost research equipment grant'',  project BALSAM, HFRI-FM17-1424. Most of this work was done while Panagiotis Patsilinakos was with the National Technical University of Athens.}} 
%

\titlerunning{Distortion of Committee Election with $1$-Euclidean Preferences and Few Distance Queries}

\author{Dimitris Fotakis\inst{1,2}  \and Laurent Gourv{\`e}s\inst{3} \and Panagiotis Patsilinakos\inst{4}}

\authorrunning{D. Fotakis, L. Gourv{\`e}s and P. Patsilinakos}

\institute{School of Electrical and Computer Engineering\\ National Technical University of Athens, 15780 Athens, Greece\\[2pt]
\and
Archimedes Research Unit, Athena RC, 15121 Athens, Greece\\[2pt]
\and 
Universit\'e Paris Dauphine-PSL, CNRS, LAMSADE, 75016, Paris, France\\[2pt]
\and
Department of Informatics, Athens University of Economics and Business, Athens 10434, Greece\\[2pt]
Emails: \email{fotakis@cs.ntua.gr}, \email{laurent.gourves@dauphine.fr}, \email{patsilinak@mail.ntua.gr}}

\maketitle              

\begin{abstract}
We consider committee election of $k \geq 2$ (out of $m \geq k+1$) candidates, where the voters and the candidates are associated with locations on the real line. Each voter's cardinal preferences over candidates correspond to her distance to the candidate locations, and each voter's cardinal preferences over committees is defined as her distance to the nearest candidate elected in the committee. We consider a setting where the true distances and the locations are unknown. We can nevertheless have access to degraded information which consists of an order of candidates for each voter. We investigate the best possible distortion (a worst-case performance criterion) wrt. the social cost achieved by deterministic committee election rules based on ordinal preferences submitted by $n$ voters and few additional distance queries. 
For $k = 2$, we achieve bounded distortion without any distance queries; we show that the distortion is $3$ for $m = 3$, and that the best possible distortion achieved by deterministic algorithms is at least $n-1$ and at most $n+1$, for any $m \geq 4$. 
For any $k \geq 3$, we show that the best possible distortion of any deterministic algorithm that uses at most $k-3$ distance queries cannot be bounded by any function of $n$, $m$ and $k$. We present deterministic algorithms for $k$-committee election with distortion of $O(n)$ with $O(k)$ distance queries and $O(1)$ with $O(k \log n)$ distance queries.
\end{abstract}

\pagenumbering{arabic}
\pagestyle{plain}

\section{Introduction}
\label{s:intro}

Electing a set of representatives based on the preferences submitted by a set of voters is a central problem in social choice. Applications span over a wide range of settings, from single and multi-winner elections to recommendation systems and machine learning (see e.g., \cite{enelow1984spatial,elkind2017multiwinner,VolkovsZ14,preflearning,LuB14}). 
%
%
In typical applications, voters express \emph{ordinal} preferences over the set of candidates, which are consistent with their \emph{cardinal} preferences, but do not include any quantitative information about the \emph{strength} of each preference. An important reason for resorting to ordinal information has to do with the cognitive difficulty of quantifying preferences. Arguably, it is much easier for a voter to rank a set of candidates from most to least preferable, than to assign an exact utility to each of them. 
However, crucial information may be lost when voters summarize cardinal to ordinal preferences. 

\citet{PR06} introduced the framework of \emph{utilitarian distortion} as a means to quantify the efficiency loss, due to the fact that elections are based on ordinal information only, and to investigate the sensitivity of voting rules to the absence of cardinal information. 
The distortion of a voting rule is the worst-case approximation ratio of its social welfare to the optimal social welfare achievable when cardinal information is available. 
Previous work has quantified the best possible distortion of single and multi-winner voting rules (often assuming normalized cardinal utilities, see e.g., \cite{BoutilierCHLPS15,CaragiannisP11,caragiannis2017subset}). 

Motivated by the frequent use of spatial preferences in social choice (see e.g., \cite{enelow1984spatial}), \citet{AnshelevichBEPS18} introduced the framework of \emph{metric distortion}, where the voters and the candidates are associated with locations in an underlying metric space. The voters' cardinal preferences over candidates correspond to their distance to the candidate locations. The voters rank the candidates in increasing order of distance and submit this information to the voting rule. Without knowledge of the voter and candidate locations and distances, the voting rule aims to minimize the sum of distances (a.k.a. the \emph{social cost}) of the voters to the candidate elected. Distortion is now defined wrt. the social cost, instead of the social welfare. In the last few years, there has been significant interest in analyzing the metric distortion of prominent voting rules (see e.g.,  \cite{AnshelevichBEPS18,SkowronE17,goel2017metric,kempe2020analysis,AnagnostidesFP22})  and in designing voting rules with optimal metric distortion for single-winner elections \cite{gkatzelis2020resolving,Kempe2023}. 

Interestingly, there has not been much previous work on the metric distortion of multiwinner voting, where we elect a committee of $k \geq 2$ (out of $m \geq k+1$) candidates based on ordinal preferences submitted by $n$ voters. As before, the voters and the candidates are associated with locations in a metric space and the voters' cardinal preferences correspond to their distance to the candidate locations. However, there are many different ways to define the voter cardinal preferences over committees, resulting in different types and desirable properties of multiwinner elections (see e.g., \cite{ElkindFSS17,FalizewksiBook2017}). 

\citet{GoelHK18} and \citet{chen2020favorite} were the first to consider the metric distortion of committee elections, in a setting where the cost of each voter for a committee is defined as the sum of her distances to all committee members. \citet{GoelHK18} proved that the best possible distortion in this setting is equal to the best possible distortion of single-winner voting and can be achieved by repeatedly applying an optimal (wrt. metric distortion) single-winner voting rule. \citet{chen2020favorite} proved that single-vote rules achieve a best possible distortion for the case where $k = m-1$, i.e., when we have to exclude a single candidate from the committee. 

\citet{CSV22} considered the metric distortion of $k$-committee election with the cost of each voter for a committee defined as her distance to the $q$-th nearest member. They proved an interesting trichotomy: the distortion is unbounded, if $q \leq k/3$, $\Theta(n)$, if $q \in (k/3, k/2]$, and equal to the best possible metric distortion of single-winner election, if $q > k/2$. For the most interesting case where $q=1$ and each voter's cost is her distance to the nearest committee member, their results imply that the distortion is $\Theta(n)$ if $k=2$, and unbounded for all $k \geq 3$, with their lower bounds standing even if the voters and the candidates are embedded in the real line. 


Subsequently, \citet{BCFRSS24,Pulya22} considered the metric distortion of classical clustering problems, such as $k$-median (which corresponds to $k$-committee election with $q=1$) and $k$-center, for $k \geq 2$, in a setting where the clustering algorithm receives only ordinal information about demand points' locations and may query few distances. They focused on the case where the voter and the candidate locations coincide (a.k.a. \emph{peer selection}\,), and asked about the minimum number of distance queries required for constant distortion. For $k$-median, \citet{Pulya22} proved that $O(1)$ distortion can be achieved deterministically with $O(n \mathrm{poly(\log n)})$ distance queries and by a randomized algorithm with $O(nk)$ queries. \citet{BCFRSS24} gave a randomized $O(1)$-distortion algorithm with $O(k^4 \log^5 n)$ queries. As for $k$-center, \citet{BCFRSS24} showed how to implement the classical $2$-approximate greedy algorithm with $k(k-1)/2$ queries and presented a deterministic $4$-distortion algorithm with only $2k$ distance queries. \citet{BCFRSS24} also proved lower bounds showing that in general metric spaces, their query bounds are not far from best possible. 

\medskip\noindent\textbf{Motivation and Objective.}
In this work, we study the metric distortion of $k$-committee elections where the cost of each voter for a committee is defined as her distance to the nearest member (i.e., we have $q=1$; it corresponds to the $k$-median setting in \cite{BCFRSS24,Pulya22}). Our setting is conceptually close to (and strongly motivated by) the prominent committee election rules of \cite{CC83,Monroe95}, which aim to elect a diverse committee that best reflects the preferences of the entire population of voters, see also \cite{ElkindFSS17,FalizewksiBook2017}. 

Our approach is rather orthogonal to \cite{BCFRSS24,Pulya22}. We consider the more general (technically more demanding and standard in computational social choice) setting where the sets (and the locations) of voters and candidates may be different, and focus on deterministic rules and on the simplest (but nevertheless interesting and challenging enough) case of \emph{$1$-Euclidean preferences}, where the voters and the candidates are embedded in the real line. As in \cite{BCFRSS24,Pulya22} (and also motivated by the success of \citet{ABFV20,ABFV22,ABFV21} in improving the utilitarian distortion for single-winner elections and one-sided matchings with cardinal queries), we aim to shed light on the following: 

\begin{question}\label{quest:main}
How many distance queries are required for a bounded (or even constant) distortion in $k$-committee election with $1$-Euclidean preferences, for $k \geq 3$? 
\end{question}

The case of $1$-Euclidean preferences is particularly interesting because it allows for a maximum possible exploitation of ordinal preferences towards achieving low distortion with a small number of distance queries. 
Moreover, the lower bounds of \cite{BCFRSS24} are based on tree metrics and do not apply to $1$-Euclidean preferences. 
As for the upper bounds of  \citet{BCFRSS24,Pulya22}, though very strong and informative about the power of distance queries in $k$-median and $k$-center, there are two key difficulties towards applying them to our setting where the voter and candidate locations do not coincide: (i) in general metric spaces, we do not know how to extract (even approximate) information about candidate-to-candidate or voter-to-voter distances from ``regular'' voter-to-candidate distance queries (the latter request information present in voter cardinal preferences); and (ii) to the best of our understanding, the algorithms of \citet{BCFRSS24,Pulya22} require ordinal information about how candidate and/or voter locations are ranked in increasing order of distance to certain candidate locations; we do not know how such ordinal information can be extracted from voter ordinal preferences, if the voter and the candidate locations are different. 

\medskip\noindent\textbf{Contribution and Techniques.}
%
We consider the general metric distortion setting, where the voter and the candidate locations may be different, focus on the simplest case of the line metric (and deal with difficulties (i) and (ii) above), and provide almost best possible answers to 
%
%
Question~\ref{quest:main}. 

For completeness, in Appendix~\ref{s:ordinal}, we start with the case where $k=2$ candidates are chosen out of $m \geq 3$ candidates arranged on the real line based on the ordinal preferences submitted by $n$ voters without using any distance queries. We show that for the case where we have $m \geq 4$ candidates, the distortion of any deterministic rule is at least $n-1$. On the positive side, we show that clustering the voters based on the two extreme candidates and electing the median of each cluster results in a distortion of at most $n+1$. We also present a deterministic rule that elects $k=2$ out of $m = 3$ candidates and achieves a distortion of $3$, which is best possible in our setting. 

Then, in Section~\ref{s:queries}, we review three different query types (voter-to-candidate, candidate-to-candidate and voter-to-voter). We show that the answer to queries of the second and the third types can be obtained from a small constant number of ``regular'' voter-to-candidate distance queries. Thus, we can rely on the more convenient candidate-to-candidate distance queries by loosing a small constant factor in the number of queries.  

In Section~\ref{s:lower-bound}, we lower bound the number of distance queries required for bounded distortion. We show that for any $k \geq 3$, the distortion of any deterministic rule that uses at most $k-3$ distance queries and selects $k$ out of $m \geq 2(k-1)$ candidates on the real line is not bounded by any function of $n$, $m$ and $k$ (Theorem~\ref{thm:query-lower-bound}). Our construction shows that a bounded distortion is not possible if we restrict distance queries to few top candidates of each voter. 

In Section~\ref{s:greedy}, we asymptotically match the lower bound above with a greedy voting rule, which uses at most $6(k-3)+3$ queries and achieves a distortion of at most $5n$ (Theorem~\ref{thm:k-queries}). It is based on the classical greedy algorithm for $k$-center \cite[Section~2.2]{WS10} (as it also happens with Polar-Opposites in \cite{CSV22} and the $k$-center algorithm in \cite[Section~3.1]{BCFRSS24}). 
In our $1$-dimensional setting, the greedy algorithm starts with the leftmost and the rightmost candidates. Then, in each iteration, it includes in the committee the furthest candidate to the set of candidates already elected. Exploiting the $1$-dimensional structure of the instance, we show how to compute the furthest candidate in each interval defined by a pair of elected candidates that are consecutive on the real line, using the voters' ordinal preferences and at most $3$ distance queries. We observe that the next candidate to be elected in the committee is one of these furthest candidates. Interestingly, greedy achieves a distortion of at most $5$ for the egalitarian cost, where we aim to minimize the maximum voter cost (and corresponds to the $k$-center objective, but with different candidate / potential center and voter / demand locations). 

In Section~\ref{s:good}, we show how to achieve low distortion with a small number of distance queries by selecting a small representative set of candidates and focusing on the restricted instance induced by them (Theorem~\ref{thm:factor3}). To demonstrate the usefulness of this reduction, in Section~\ref{s:coresets}, we exploit a generalization of the greedy rule. Our construction for selecting a small representative set of candidates is inspired by the notion of \emph{coresets}, extensively used for $k$-median in computational geometry (see e.g., \cite{FrahlingS05}). Our construction uses $O(k\log n)$ distance queries and computes a set of $O(k\log n)$ representative candidates that allow for a distortion of $5$ (Theorem~\ref{thm:coreset}).
The idea is to maintain a hierarchical partitioning of the candidate axis into a set of intervals, so that we can upper bound the contribution of the voters associated with each interval to the social cost. In each iteration, the most expensive interval, wrt. its contribution to the social cost, is split into two subintervals, defined by the two candidates on the left and on the right of the interval's midpoint. The key step is to show that interval subdivision can be implemented using the voter ordinal preferences and at most $4$ distance queries.

\smallskip\noindent\textbf{Related Work.}
Metric distortion was introduced in \cite{AnshelevichBEPS18}, where the distortion of many popular voting rules for single-winner elections was studied. Subsequent work analyzed the metric distortion of popular voting rules, such as STV \cite{SkowronE17,AnagnostidesFP22}. \citet{MunagalaW19} and \citet{kempe2020analysis} presented deterministic rules with distortion $2+\sqrt{5}$, breaking the barrier of $5$ achieved by Copeland. \citet{gkatzelis2020resolving} introduced Plurality Matching and proved that it achieves an optimal distortion of $3$ in general metric spaces (see also \cite{Kempe2023}). \citet{AnshZ21} studied the distortion of single and multiwinner elections with known candidate locations. 
\citet{Abramowitz2019} resorted to additional information about the strength of voter preferences in order to improve on the best known distortion for single-winner elections. 
The reader is referred to the survey of \citet{Anshelevich2021} for a detailed overview. 

\citet{BoutilierCHLPS15} and \citet{caragiannis2017subset} studied the best possible utilitarian distortion of single and multiwinner elections, respectively. \citet{ABFV20}  significantly improved on the best possible utilitarian distortion for single-winner elections using cardinal information. They introduced a family of single-winner voting rules with distortion $O(m^{1/(\ell+1)})$ using $O(n \ell\log m)$ value queries. Subsequently, \citet{ABFV22,ABFV21} significantly improved on the best possible utilitarian distortion for one-sided matchings using algorithms that resort to a small number of value queries per voter. Interestingly, our query bounds are linear in the size $k$ of the committee and only logarithmic in the number $n$ of voters (instead of linear in $n$ in the query bounds of \cite{ABFV20,ABFV22,ABFV21}).

The main result of \cite{FGM16} implies that the metric distortion for the utilitarian version of $k$-committee election with $1$-Euclidean preferences is $3$, for $k \in \{1, 2\}$, and at most $\frac{2k-1}{2k-3}$, for any $k \geq 3$. Namely, the distortion for the utilitarian version of $k$-committee election on the real line tends to $1$ as the committee size $k$ increases. 

Multiwinner voting is a significant research direction in social choice and has been studied from many different viewpoints, e.g., proportional representation \cite{AzizBCEFW17,PetersS20}, axiomatic justification \cite{ElkindFSS17}, core-stability in restricted domains \cite{PierczynskiS22}.  Selection of a single candidate or a committee of candidates based on $1$-Euclidean preferences submitted by voters (or agents) is a typical setting in social choice and mechanism design and has been the topic of previous work (see e.g., \cite{Miy01,ProcacciaT13,FotakisT14,FeldmanFG16} for representative previous work on mechanism design, and \cite{FGM16,FG22} and few references in \cite{Anshelevich2021} for representative previous work on social choice and distortion).

\section{Model and Notation}
\label{s:notation}

%
%
We consider a set $\C = \{ c_1, \ldots, c_m \}$ of $m$ candidates and a set $\V = \{ v_1, \ldots, v_n \}$ of $n$ voters. We assume that they are all located on the real line $\reals$, i.e., each candidate $c_i$ (resp.
voter $v_j$) is associated with a location $x(c_i) \in \reals$ (resp. $x(v_j) \in \reals$). For brevity, we usually let $c_i$ (resp. $v_j$) denote both the candidate (resp. the voter) and her location $x(c_i)$ (resp. $x(v_j)$). We always index candidates in increasing order of their real coordinates, i.e., $c_1 < c_2 < \cdots < c_m$, which is also the order they appear on the candidate axis from left to right. We let $\C[c, c'] = \C \cap [c, c']$ be the set (or interval) of candidates in $\C$ between $c$ and $c'$ on the candidate axis.

For each voter $v$, we let her $L_1$ distance to the candidate locations quantify her cardinal preferences over $\C$. I.e., $v$'s cost for being represented by a candidate $c$ is
\[ \cost_v(c) = d(v, c) = |v - c| = |x(v) - x(c)|\,. \]
For a voter $v$ and a set $S \subseteq \C$ of candidates, we let 
$d(v, S) = \min_{c \in S} \{ d(v, c) \}
           = \min_{c \in S} \{ |v - c| \}$.
%
Motivated by the \citet{CC83} rule for $k$-committee election, we assume that each voter $v$ is represented by (or is assigned to) her nearest candidate in any given set $S$ of elected candidates. 
Formally, for any $S \subseteq \C$, we let 
%
%
$\cost_v(S) = d(v, S) = \min_{c \in S} \{ d(v, c) \}$ be the cost experienced by $v$ from the set $S$ of elected candidates.

\smallskip\noindent\textbf{Problem Definition.}
The problem of \textbf{$k$-Committee Election} is to select a candidate set (a.k.a. committee) $S \subseteq \C$, with $|S| = k \leq m-1$, that minimizes the (utilitarian) \emph{social cost} $\Util(S) = \sum_{v \in \V} \cost_v(S)$ of the voters. We also consider the \emph{egalitarian cost} $\Egal(S) = \max_{v \in \V} \cost_v(S)$ of the voters for a $k$-committee $S$ of elected candidates. 
We often refer to $(\C, \V)$, where $\C$ is the set of candidates and $\V$ is the set of voters, along with their locations on the real line (which are assumed fixed, but unknown to the voting rule), as an \emph{instance} of $k$-Committee Election. 

 \smallskip\noindent\textbf{Committee Election with 1-Euclidean Preferences and Distance Queries.}
$k$-Committee Election can be solved in $O(n k \log n)$ time, by dynamic programming \cite{kmedian}, if we have access to the voter and the candidate locations on the real line (or to all voter-candidate distances). However, in our setting, every voter $v$ provides only a ranking $\succ_v$ over the set $\C$ of candidates that is consistent with the function $\cost_v : \C \to \reals_{\geq 0}$. Namely, for every two candidates $c$ and $c'$, $c \succ_v c'$ (i.e., $v$ prefers $c$ to $c'$) if and only if $d(v, c) < d(v, c')$. Every voter's ranking constitutes her preference. As usual in relevant literature (see e.g., \cite[Section~2]{Anshelevich2021}), we assume that for every voter $v$, $\succ_v$ is a strict total order, i.e., that for every pair of candidates $c$ and $c'$, $d(v, c) \neq d(v, c')$.

Our committee election rules receive a ranking profile $\vec{\succ} = (\succ_1, \ldots, \succ_n)$ consisting of a strict total order $\succ_j$ over $\C$ for each voter $v_j \in \V$. 
We only consider \emph{$1$-Euclidean} ranking profiles $\vec{\succ}$, in the sense that all $\succ_v$ in $\vec{\succ}$ are  consistent with a cost function $\cost_v$ computed wrt. some fixed (and common) collection of voter and candidate locations on the real line. 
Under the assumption that total orders $\succ_v$ are strict, $1$-Euclidean ranking profiles are \emph{single-peaked} \cite{DB48} and \emph{single-crossing} \cite{K68,M71}, properties that have received significant attention in computational social choice (see e.g., \cite{ELO08,EF14} and the references therein). 
\citet[Section~3]{FG22} present simple examples, where candidates and voters are embedded in the real line and different tie breaking in voter ordinal preferences results in ranking profiles that are not single-peaked or single-crossing.
\citet{EF14} show that given a ranking profile $\vec{\succ}$, we can verify if $\vec{\succ}$ is $1$-Euclidean and compute in polynomial time a strict ordering of the candidates on the real line, from left to right, that is consistent with $\vec{\succ}$. We refer to such an ordering as the \emph{candidate axis}. 


A \emph{deterministic rule} $R$ for $k$-committee election receives a $1$-Euclidean ranking profile $\vec{\succ} = (\succ_1, \ldots, \succ_n)$ over a set $\C$ of $m$ candidates, the desired committee size $k$ and a non-negative integer $q$. Then, using $\vec{\succ}$ and information about the distance of at most $q$ candidate pairs on the real line, $R$ computes a committee $R(\vec{\succ}, k, q) = S \subseteq \C$ with $k$ candidates. Our committee election rules assume availability of the candidate axis corresponding to $\vec{\succ}$ and may ask up to $q$ distance queries \emph{adaptively}%
\footnote{Namely, the queries are not fixed in advance. A query may depend on the answer to previous queries.}.
We assume that the responses to all distance queries are consistent with a fixed collection of voter and candidate locations on $\reals$ that result in $\vec{\succ}$. 

\smallskip\noindent\textbf{Distortion.}
We evaluate the performance of a committee election rule $R$ (often called \emph{rule} or \emph{algorithm}, for brevity) for given ranking profile $\vec{\succ}$, committee size $k$ and query number $q$ in terms of its \emph{distortion} \cite{BoutilierCHLPS15,PR06}, i.e, the worst-case approximation ratio that $R$ achieves wrt. the social cost:
\begin{equation}\label{eq:distortion}
 \dist(R, \vec{\succ}, k, q) = \sup \frac{\Util(R(\vec{\succ}, k, q))}{\min_{S: |S| = k} \Util(S)}\,,
\end{equation}
where the supremum is taken over all collections of voter and candidate locations on the real line that are consistent with $\vec{\succ}$ and with the responses to the $q$ candidate distance queries asked by $R$. 
The distortion of a deterministic $k$-committee rule $R$ is the maximum of $\dist(R, \vec{\succ}, k, q)$ over all linear ranking profiles $\vec{\succ}$ with $n$ voters and $m$ candidates. We sometimes also consider the distortion wrt. the egalitarian cost $\Egal(S)$ by explicitly referring to it.  

\smallskip\noindent\textbf{Notation.}
We let $\top(v)$ be the top candidate of voter $v$ in $\succ_v$. A candidate's $c$ cluster $\Cl(c)$ consists of all voters in $\V$ with $c$ as their top candidate. We say that a candidate $c$ is \emph{active} if $\Cl(c) \neq \emptyset$, i.e., there is some voter $v$ with $c$ as her top choice. We assume \emph{non-degenerate} ranking profiles $\vec{\succ}$, where $n \geq k+1$ and all candidates are located between the leftmost and the rightmost active candidate. We justify this assumption in Appendix~\ref{s:app:top-choice}. We note that the candidate axis, determined from $\vec{\succ}$ by the algorithm of \cite{EF14}, is guaranteed to be unique (up to symmetries) for non-degenerate ranking profiles $\vec{\succ}$. 

We always assume that the candidate set $\C$ given as input to our algorithms  consists of \emph{active candidates only}. An instance is \emph{candidate-restricted}, if all candidates are active and all voters are moved to the location of their top candidate. Hence, a candidate-restricted instance can be represented with $m$ pairs $(c_i,n_i)$ where $n_i$ is the number of voters co-located with candidate $c_i$.
%
%
Assuming that all voters are collocated with their top candidates (and then removing inactive candidates, see Appendix~\ref{s:app:active-candidates}) increases the distortion by a factor of at most $3$ (see Theorem~\ref{thm:factor3} for the social cost and Appendix~\ref{s:app:good-egalitarian} for the egalitarian cost). 

By stating that explicitly, the analysis of our distortion bounds sometimes  uses candidate-restricted instances. We should highlight that our algorithms work without assuming anything about voter and candidate locations (except that candidate and voter locations lie on the real line) and our distortion bounds hold against an optimal solution for the original instance, where some candidates may be inactive, and candidate and voter locations may be different. 

There is a delicate issue that restricts the use of candidate-restricted instances in our algorithms (and with which our algorithms carefully deal): 
When we use a ranking $\succ_v$ in an algorithm, we have to take care of the fact that $\succ_v$ may be different from the ranking $\succ_{\top(v)}$, where the candidates are ranked in increasing order of distance to $\top(v)$ (because the locations of $v$ and $\top(v)$ may be different). The difficulty of deducing useful information about the rankings $\succ_{c}$ at candidate locations $c$ from a voter ranking profile $\vec{\succ}$ imposes a significant difference between our setting and the clustering setting in \cite{BCFRSS24,Pulya22}.


\section{Distance Query Types}
\label{s:queries}

Before moving on to analyze the distortion of committee election rules that can use a small number of distance queries, we discuss three different types of them: 

\begin{description}
\item[Regular queries.] Given a voter $v \in \V$ and a candidate $c \in \C$, we ask for $d(v, c) = |v - c|$. 

\item[Candidate queries.] Given two candidates $c, c' \in \C$, we ask for the distance $d(c, c') = |c - c'|$.

\item[Voter queries.] Given two voters $v, v' \in \V$, we ask for the distance $d(v, v') = |v - v'|$.
\end{description}

Regular queries ask for information available in the voter cost functions $\cost_v : \C \to \reals_{\geq 0}$ quantifying their cardinal utilities. In Appendix~\ref{s:app:queries}, we show how to simulate candidate queries and voter queries with at most $6$ and $2$ regular queries, respectively, for the $1$-Euclidean case.


Therefore, as long as we care about the asymptotics of the number of queries used by a committee election rule, we may use these types of queries interchangeably. Hence, we state and analyze our committee election rules assuming access to candidate queries, with the understanding that they can be implemented using asymptotically the same number of regular queries. The exact number of regular queries can be improved via a careful implementation that uses regular queries directly.


\section{Constant Distortion with $\Theta(m)$ Queries}
\label{s:app:axis}

We next observe that for any $k \geq 2$ and any $m \geq k+1$, we can fully reconstruct the candidate axis using $m-1$ candidate distance queries and find the optimal $k$-committee for the corresponding candidate-restricted instance using dynamic programming. This implies a distortion of $3$ (for both the social and the egalitarian cost) using $m-1$ distance queries.

\begin{observation}[Reconstructing the Candidate Axis]\label{obs:axis}
The distances between all pairs of consecutive active candidates on the real line can be extracted from $m-1$ candidate (or from $m+3$ regular) 
distance queries. 
\end{observation}

\begin{proof}
Let $c_1 < c_2 < \cdots < c_m$ be the active candidates as they appear on the real line from left to right. 
If we have access to candidate queries, we query the distances $d(c_1, c_2)$, $d(c_2, c_3)$, $\ldots$, $d(c_{m-1}, c_m)$. 
%
%
Otherwise, we let $v_1 \in \Cl(c_1)$ (resp. $v_m \in \Cl(c_m)$) be any voter with $\top(v_1) = c_1$ (resp. $\top(v_m) = c_m$). Working as in the proof of 
Proposition~\ref{pr:candidate_queries}, we use $6$ regular queries in order to determine how $v_1$ (resp. $v_m$) is located relatively to $c_1$ (resp. $c_m$) and the distance $d(c_1, c_m)$. Then, we query the distances $d(v_1, c_3), \ldots, d(v_1, c_{m-1})$ (we have already queried the distances $d(v_1, c_2)$ and $d(v_1, c_m)$ in the first $6$ queries). We let either $d(c_1, c_2) = d(v_1, c_2)-d(v_1, c_1)$ or $d(c_1, c_2) = d(v_1, c_1)+d(v_1, c_2)$ depending on whether $v_1 < c_1$ or $c_1 \leq v_1$. For each $c_i$, $3 \leq i \leq m$, we let $d(c_{i-1}, c_i) = d(v_1, c_i) - d(v_1, c_{i-1})$. Hence, we conclude that these $m+3$ distances provide all the information required to compute the distances $d(c_{i-1}, c_{i})$, $2 \leq i \leq m$, for all pairs $c_{i-1}$ and $c_{i}$ of consecutive candidates on the real line. 
\qed\end{proof}

We note that as soon as we have the distances between all consecutive active candidates, an optimal $k$-committee for the candidate-restricted instance $\C_{\cres}$ induced by $\C$ can be computed in polynomial time by dynamic programming \cite{kmedian}. Then, Proposition~\ref{pr:active} in Appendix~\ref{s:app:active-candidates} and Theorem~\ref{thm:factor3} in Section~\ref{s:good} (resp. Theorem~\ref{thm:factor3_egal} in Appendix~\ref{s:app:good-egalitarian}) imply that $m-1$ candidate (or $m+3$ regular) distance queries are sufficient for a distortion of at most $3$, for any $k \geq 2$ and any $m \geq k+1$ for both the utilitarian cost and the egalitarian cost.



\section{Lower Bound on the Number of Queries Required for Bounded Distortion}
\label{s:lower-bound}

We next show that when we select $k \geq 3$ out of $m \geq 2(k-1)$ candidates, achieving a bounded distortion requires at least $k-2$ distance queries. 

\begin{theorem}\label{thm:query-lower-bound}
For any $k \geq 3$, the distortion of any deterministic $k$-committee election rule that uses at most $k-3$ distance queries and selects $k$ out of at least $2(k-1)$ candidates on the real line cannot be bounded by any function of $n$, $m$ and $k$ (for both the social cost and the egalitarian cost). 
\end{theorem}

\begin{proof}
For every $k \geq 3$, we construct a family of $2(k-1)$ instances on the real line with $k-1$ candidate pairs each that cannot be distinguished with less than $k-2$ distance queries. The distances are chosen so that in any committee with a bounded distortion, both candidates of a particular pair must be chosen, along with one candidate from each of the remaining pairs. Any deterministic rule cannot tell that particular pair, unless it asks at least $k-2$ distance queries in the worst case. 

For the construction, we consider $m = 2(k-1)$ candidates, $c_1 < c_2 < \cdots < c_{2k-3} < c_{2k-2}$. We let $D$ sufficiently large, so that $D^2 \gg \max\{2D+1, k\}$, and an $\epsilon \in (0, 1/k)$ sufficiently small used for tie breaking. In the basic instance, we let $d(c_{2i-1}, c_{2i}) = 1$, for all $i \in [k-1]$, and let $d(c_{2i}, c_{2i+1}) = D^2 + (i-1)\epsilon$, for all $i \in [k-2]$ (see also Figure~\ref{fig:basic_lb}). There are $n = m$ voters, each with a different top candidate. The voters are collocated with their top candidate in the basic instance and its variants presented below.

\begin{figure*}[t]
\begin{align*}
&\hspace*{-1.2em}%
   c_1 \stackrel{1}{\xdash[1em]} c_2 \stackrel{D^2}{\xdash[6em]} 
   c_3 \stackrel{1}{\xdash[1em]} c_4 \stackrel{D^2+\eps}{\xdash[6em]} 
   c_5 \stackrel{1}{\xdash[1em]} c_6 \stackrel{D^2+2\eps}{\xdash[6em]}
   c_7 \stackrel{1}{\xdash[1em]} c_8 \stackrel{D^2+3\eps}{\xdash[6em]} 
   c_9 \stackrel{1}{\xdash[1em]} c_{10}
\end{align*}
\caption{The basic instance used in the lower bound of Theorem~\ref{thm:query-lower-bound} for $k=6$.}\label{fig:basic_lb}
\end{figure*}

\begin{figure*}[th]
\begin{align*}
&\hspace*{-1.2em}%
   c_1 \stackrel{D+1}{\xdash[2.2em]} c_2 \stackrel{D^2}{\xdash[6em]} 
   c_3 \stackrel{1}{\xdash[1em]} c_4 \stackrel{D^2+\eps}{\xdash[6em]} 
   c_5 \stackrel{1}{\xdash[1em]} c_6 \stackrel{D^2+2\eps}{\xdash[6em]}
   c_7 \stackrel{1}{\xdash[1em]} c_8 \stackrel{D^2+3\eps}{\xdash[6em]} 
   c_9 \stackrel{1}{\xdash[1em]} c_{10} \\
&  c_1 \stackrel{D+1}{\xdash[2.2em]} c_2 \stackrel{D^2-D}{\xdash[5em]} 
   c_3 \stackrel{1}{\xdash[1em]} c_4 \stackrel{D^2+\eps}{\xdash[6em]} 
   c_5 \stackrel{1}{\xdash[1em]} c_6 \stackrel{D^2+2\eps}{\xdash[6em]}
   c_7 \stackrel{1}{\xdash[1em]} c_8 \stackrel{D^2+3\eps}{\xdash[6em]} 
   c_9 \stackrel{1}{\xdash[1em]} c_{10} \\
&  c_1 \stackrel{1}{\xdash[1em]} c_2 \stackrel{D^2-D}{\xdash[5em]} 
   c_3 \stackrel{D+1}{\xdash[2.2em]} c_4 \stackrel{D^2+\eps}{\xdash[6em]} 
   c_5 \stackrel{1}{\xdash[1em]} c_6 \stackrel{D^2+2\eps}{\xdash[6em]}
   c_7 \stackrel{1}{\xdash[1em]} c_8 \stackrel{D^2+3\eps}{\xdash[6em]} 
   c_9 \stackrel{1}{\xdash[1em]} c_{10} \\
&  c_1 \stackrel{1}{\xdash[1em]} c_2 \stackrel{D^2}{\xdash[6em]} 
   c_3 \stackrel{D+1}{\xdash[2.2em]} c_4 \stackrel{D^2-D+\eps}{\xdash[5em]} 
   c_5 \stackrel{1}{\xdash[1em]} c_6 \stackrel{D^2+2\eps}{\xdash[6em]}
   c_7 \stackrel{1}{\xdash[1em]} c_8 \stackrel{D^2+3\eps}{\xdash[6em]} 
   c_9 \stackrel{1}{\xdash[1em]} c_{10} \\
&  c_1 \stackrel{1}{\xdash[1em]} c_2 \stackrel{D^2}{\xdash[6em]} 
   c_3 \stackrel{1}{\xdash[1em]} c_4 \stackrel{D^2-D+\eps}{\xdash[5em]} 
   c_5 \stackrel{D+1}{\xdash[2.2em]} c_6 \stackrel{D^2+2\eps}{\xdash[6em]}
   c_7 \stackrel{1}{\xdash[1em]} c_8 \stackrel{D^2+3\eps}{\xdash[6em]} 
   c_9 \stackrel{1}{\xdash[1em]} c_{10} \\
&  c_1 \stackrel{1}{\xdash[1em]} c_2 \stackrel{D^2}{\xdash[6em]} 
   c_3 \stackrel{1}{\xdash[1em]} c_4 \stackrel{D^2+\eps}{\xdash[6em]} 
   c_5 \stackrel{D+1}{\xdash[2.2em]} c_6 \stackrel{D^2-D+2\eps}{\xdash[5em]}
   c_7 \stackrel{1}{\xdash[1em]} c_8 \stackrel{D^2+3\eps}{\xdash[6em]} 
   c_9 \stackrel{1}{\xdash[1em]} c_{10} \\
&  c_1 \stackrel{1}{\xdash[1em]} c_2 \stackrel{D^2}{\xdash[6em]} 
   c_3 \stackrel{1}{\xdash[1em]} c_4 \stackrel{D^2+\eps}{\xdash[6em]} 
   c_5 \stackrel{1}{\xdash[1em]} c_6 \stackrel{D^2-D+2\eps}{\xdash[5em]}
   c_7 \stackrel{D+1}{\xdash[2.2em]} c_8 \stackrel{D^2+3\eps}{\xdash[6em]} 
   c_9 \stackrel{1}{\xdash[1em]} c_{10} \\
&  c_1 \stackrel{1}{\xdash[1em]} c_2 \stackrel{D^2}{\xdash[6em]} 
   c_3 \stackrel{1}{\xdash[1em]} c_4 \stackrel{D^2+\eps}{\xdash[6em]} 
   c_5 \stackrel{1}{\xdash[1em]} c_6 \stackrel{D^2+2\eps}{\xdash[6em]}
   c_7 \stackrel{D+1}{\xdash[2.2em]} c_8 \stackrel{D^2-D+3\eps}{\xdash[5em]} 
   c_9 \stackrel{1}{\xdash[1em]} c_{10} \\
&  c_1 \stackrel{1}{\xdash[1em]} c_2 \stackrel{D^2}{\xdash[6em]} 
   c_3 \stackrel{1}{\xdash[1em]} c_4 \stackrel{D^2+\eps}{\xdash[6em]} 
   c_5 \stackrel{1}{\xdash[1em]} c_6 \stackrel{D^2+2\eps}{\xdash[6em]}
   c_7 \stackrel{1}{\xdash[1em]} c_8 \stackrel{D^2-D+3\eps}{\xdash[5em]} 
   c_9 \stackrel{D+1}{\xdash[2.2em]} c_{10} \\
&  c_1 \stackrel{1}{\xdash[1em]} c_2 \stackrel{D^2}{\xdash[6em]} 
   c_3 \stackrel{1}{\xdash[1em]} c_4 \stackrel{D^2+\eps}{\xdash[6em]} 
   c_5 \stackrel{1}{\xdash[1em]} c_6 \stackrel{D^2+2\eps}{\xdash[6em]}
   c_7 \stackrel{1}{\xdash[1em]} c_8 \stackrel{D^2+3\eps}{\xdash[6em]} 
   c_9 \stackrel{D+1}{\xdash[2.2em]} c_{10}
\end{align*}
\caption{The $2(k-1) = 10$ variants obtained from the basic instance used in the lower bound of Theorem~\ref{thm:query-lower-bound} for $k=6$.}\label{fig:variants_lb}
\end{figure*}

We construct a family of $2(k-1)$ different variants of the basic instance, by moving candidate $c_j$, $j = 1, \ldots, 2(k-1)$, by $D$, while keeping every other candidate at her original location. Specifically, in the $j$-th variant, if $j$ is odd, we increase the distance $d(c_{j}, c_{j+1})$ from $1$, in the basic instance, to $D+1$, by moving candidate $c_{j}$ by $D$ on the left. In the $j$-th variant, if $j$ is even, we increase the distance $d(c_{j-1}, c_{j})$ from $1$, in the basic instance, to $D+1$, by moving candidate $c_{j}$ by $D$ on the right. All other candidates maintain the locations that they have in the basic instance (see also Figure~\ref{fig:variants_lb}). 
As a result, in the $j$-th variant, if $j$ is odd, all distances $d(c_i, c_{j})$, for $i = 1, \ldots, j-1$, decrease by $D$, while all distances $d(c_{j}, c_i)$, for $i = j+1, \ldots, 2(k-1)$, increase by $D$. If $j$ is even, all distances $d(c_i, c_{j})$, for $i = 1, \ldots, j-1$, increase by $D$, while all distances $d(c_{j}, c_i)$, for $i = j+1, \ldots, 2(k-1)$, decrease by $D$. These changes affect the distance of candidate $c_{j}$ to all other candidates, but do not affect the distances between other candidate pairs, which remain as in the basic instance (see also Example~\ref{ex:lower_bound}).

\begin{example}\label{ex:lower_bound}
For $k = 6$, the basic instance used in our construction is shown in Figure~\ref{fig:basic_lb} and the $2(k-1) = 10$ variants obtained from the basic instance are shown in Figure~\ref{fig:variants_lb}.

In the basic instance and its variants, the preference list (the same for all of them) of each of the $n=10$ voters collocated with the $m=10$ candidates is:
\begin{enumerate}
\item $(c_1, c_2, c_3, c_4, c_5, c_6, c_7, c_8, c_9, c_{10})$
\item $(c_2, c_1, c_3, c_4, c_5, c_6, c_7, c_8, c_9, c_{10})$

\item $(c_3, c_4, c_2, c_1, c_5, c_6, c_7, c_8, c_9, c_{10})$
\item $(c_4, c_3, c_5, c_2, c_6, c_1, c_7, c_8, c_9, c_{10})$

\item $(c_5, c_6, c_4, c_3, c_7, c_8, c_2, c_1, c_9, c_{10})$
\item $(c_6, c_5, c_7, c_4, c_8, c_3, c_9, c_2, c_{10}, c_1)$

\item $(c_7, c_8, c_6, c_5, c_9, c_{10}, c_4, c_3, c_2, c_1)$ 
\item $(c_8, c_7, c_9, c_6, c_{10}, c_5, c_4, c_3, c_2, c_1)$

\item $(c_9, c_{10}, c_8, c_7, c_6, c_5, c_4, c_3, c_2, c_1)$
\item $(c_{10}, c_9, c_8, c_7, c_6, c_5, c_4, c_3, c_2, c_1)$
\end{enumerate}
We note that the ranking of each voter in each variant is identical to her ranking over the candidates in the basic instance. Therefore, a voting rule cannot identify the particular variant by only considering the rankings over candidates submitted by the voters. 

Moreover, the distance change in the $j$-th variant only affects the distance of candidate $c_j$ to the remaining candidates. E.g., for $j=1$, all distances $d(c_1, c_i)$, $i = 2, \ldots, 10$ increase by $D$. For $j=2$, the distance $d(c_1, c_2)$ increases by $D$, while the distances $d(c_2, c_i)$, $i =3, \ldots, 10$ decrease by $D$. In general, in the $j$-th variant, for $j = 1, \ldots, 10$, if $j$ is odd, the distances $d(c_i, c_{j})$, for $i = 1, \ldots, j-1$, decrease by $D$, while all distances $d(c_{j}, c_i)$, for $i = j+1, \ldots, 10$, increase by $D$. Symmetrically, if $j$ is even, all distances $d(c_i, c_{j})$, for $i = 1, \ldots, j-1$, increase by $D$, while all distances $d(c_{j}, c_i)$, for $i = j+1, \ldots, 2(k-1)$, decrease by $D$. 
\qed\end{example}

Intuitively, the basic instance (and each variant) consists of $k-1$ essentially isolated candidate pairs. In each variant, the candidates of exactly one pair are far away from each other (so for a bounded distortion, we need to identify this pair and elect both candidates), while the candidates of the remaining pairs are quite close to each other (so we may elect any of them). Since the $2(k-1)$ variants are symmetric otherwise, any distance query that discovers that the distance of a candidate pair is as in the basic instance can exclude at most two variants (from the list of all possible instances used in this proof). Therefore, any deterministic rule requires at least $k-2$ distance queries in the worst case, before it is able to identify the candidate pair at distance $D$ to each other. 


More formally, the optimal committee (for both the social cost and the egalitarian cost) for the $j$-th variant of the basic instance is to select the candidates $c_{j}$ and $c_{j+1}$, if $j$ is odd, and $c_{j-1}$ and $c_j$, if $j$ is even, which candidates are at distance $D$ to each other, and any candidate from each of the remaining candidate pairs $(c_{2i-1}, c_{2i})$. The social cost of the optimal committee is $k-2$ (and the egalitarian cost of the optimal committee is $1$). Any other committee for the $j$-th variant has social (resp. egalitarian) cost at least $D$, which can become arbitrarily larger than $k-2$ (resp. than $1$). Therefore, a deterministic committee election rule with a bounded distortion must be able to identify the candidate pair at distance $D$ (or equivalently, to identify the right variant) and to add both these candidates to the remaining $k-2$ candidates elected by the rule. 

We next show that this is not possible, unless the algorithm asks for at least $k-2$ distance queries. We first observe that the ranking of any voter in the basic instance is identical to her ranking over candidates in each of the $2(k-1)$ variants obtained from the basic instance (see also Example~\ref{ex:lower_bound}). Therefore, a voting rule cannot tell the right variant by only looking at the rankings submitted by the voters. 

We also observe that in the $j$-th variant, the distances of candidate $c_{j}$ to all other candidates differ from the corresponding distances in the basic instance, but the distances between all other candidate pairs are as in the basic instance. Therefore, every time we query the distance between a pair of candidates and find it unchanged with respect to the basic instance, we can exclude at most two variants from the family of $2(k-1)$ variants defined above. Consequently, for any $k \geq 3$, any deterministic committee election rule needs at least $k-2$ distance queries in the worst case, before it is able to identify the pair of candidates $(c_{2i-1}, c_{2i})$ that are at distance $D$ in the input variant (and we need to elect both). 
\qed\end{proof}

An interesting question is if bounded distortion is possible with distance queries restricted to the few top candidates of each voter. In the proof of Theorem~\ref{thm:query-lower-bound}, we can embed a large group of candidates, located extremely close to each other, in each location $c_i$ in the basic instance and its variants (and have a voter collocated with each of them). The candidates of each group are ranked first by all voters in the same group. Hence, unless we are allowed to query distances to candidates further down the voters' rankings, we cannot get any useful information about these instances.

\section{Bounded Distortion with $\Theta(k)$ Queries}
\label{s:greedy}

In this section, we present a simple greedy rule for $k$-committee election that achieves bounded distortion with $\Theta(k)$ distance queries, thus asymptotically matching the lower bound of Theorem~\ref{thm:query-lower-bound} wrt. the number of queries required for bounded distortion. 

\begin{algorithm}[t]
\caption{\label{alg:k-center}The greedy algorithm for $k$-committee election}
\textbf{Input}: Candidates $\C = \{ c_1, \ldots, c_m \}$, $k \in \{2, \ldots, m-1\}$, distance function $d:\C \times \C \to \reals_{\geq 0}$. \\
\textbf{Output}: Set $S \subseteq \C$ of $k$ candidates.

\begin{algorithmic}[1]
\STATE $S \leftarrow \{ c_1, c_m \}$ \COMMENT{pick leftmost and rightmost candidates}
\WHILE{$|S| < k$}
    \STATE $\hat c \leftarrow \arg\max_{c \in \C} \{ d(c, S) \}$
    \STATE $S \leftarrow S \cup \{ \hat c \}$
\ENDWHILE
\RETURN $S$
\end{algorithmic}
\end{algorithm}

We show that the classical $2$-approximate greedy algorithm for $k$-center can be implemented with few distance queries. The Greedy algorithm \cite[Section~2.2]{WS10} iteratively maintains a set $S$ of candidates, starting with any candidate, and adding the candidate $c$ with maximum distance $d(c, S)$ to the current set $S$ in each iteration. When applied to $1$-Euclidean instances, Algorithm~\ref{alg:k-center} starts with the leftmost candidate $c_1$ and the rightmost candidate $c_m$. Then, for the next $k-2$ iterations, it adds to $S$ the candidate $c\in\C$ with maximum $d(c, S)$.



To implement Algorithm~\ref{alg:k-center} with distance queries (see Algorithm~\ref{alg:greedy}), we need to compute the most distant candidate in $\C$ to current candidate set $S = \{ c_1, \ldots, c_\ell \}$, while $\ell < k$. For convenience, we let the candidates in $S$ be indexed as they appear on the candidate axis, from left to right, i.e., $c_1 < c_2 < \cdots < c_\ell$, and $c_1$ (resp. $c_\ell$) is the leftmost (resp. rightmost) candidate in $S$. 


Algorithm~\ref{alg:greedy} maintains a set $\hat{C}$ with $\ell-1$ candidate-distance pairs $(\hat{c}_i, \delta_i)$, where for each $i \in [\ell-1]$, $\hat{c}_i$ is the most distant candidate in the interval $\C[c_i, c_{i+1}]$ to its endpoints $c_i, c_{i+1} \in S$ and $\delta_i = d(\hat{c}_i, \{ c_i, c_{i+1} \})$ is its distance to the endpoints of $\C[c_i, c_{i+1}]$. More formally, for each $i \in [\ell-1]$, we let 
\begin{equation}\label{eq:most-distant}
\hat{c}_i = \arg\max_{c \in \C[c_i, c_{i+1}]} \big\{ d(c, \{ c_i, c_{i+1} \} ) \big\} \mbox{\ \ \ and\ \ \ } \delta_i = d(\hat{c}_i, \{ c_i, c_{i+1} \})\,.
\end{equation}

This information is provided by the Distant-Candidate algorithm (Algorithm~\ref{alg:dist-cand}). Every time a new candidate $c$, lying between $c_i$ and $c_{i+1} \in S$ on the axis, is added to $S$, the most distant candidates $\hat{c_i}$, together with its distance $\delta_i$ to $\{c_i, c\}$, and $\hat{c}_{i+1}$, together with its distance $\delta_{i+1}$ to $\{c, c_{i+1}\}$, are computed by two calls to the Distant-Candidate algorithm and are added to $\hat{C}$ (Algorithm~\ref{alg:greedy}, step~10), and the pair $(c, \delta)$ corresponding to $c$ is removed from $\hat{C}$ (step~6). 

In the next iteration, the most distant candidate in $\C$ to the current set $S$ is computed (step~4) and added to $S$ (step~5). For step~4, we observe that the most distant candidate to $S$ corresponds to the pair $(c, \delta) \in \hat{C}$, where $c$ maximizes the distance $\delta$ to its neighbor candidates in $S$ among all $(c', \delta') \in \hat{C}$. This observation is formalized by the following: 

\begin{algorithm}[t]
\caption{\label{alg:greedy}Query-efficient implementation of the greedy algorithm}
\textbf{Input}: Candidates $\C = \{ c_1, \ldots, c_m \}$, $k \in \{2, \ldots, m-1\}$, voter ranking profile $\vec{\succ} = (\succ_1, \ldots, \succ_n)$\\
\textbf{Output}: Set $S \subseteq \C$ of $k$ candidates

\begin{algorithmic}[1]
\STATE $S \leftarrow \{ c_1, c_m \}$ \COMMENT{pick leftmost and rightmost candidates}
\STATE $\hat{C} \leftarrow \big\{ \text{Distant-Candidate}(\C[c_1, c_m]) \big\}$
\WHILE{$|S| < k$}
	\STATE Let $c$ be s.t. $(c,\delta)\in \hat{C}$ and $\delta \ge \delta'$ for all $(c',\delta') \in \hat{C}$
	\STATE $S \leftarrow S \cup \{ c \}$
    \STATE $\hat{C} \leftarrow \hat{C} \setminus \{ (c, \delta) \}$
    \IF{$|S| < k$}
		\STATE Let $c_i$ be the rightmost candidate in $S$ on $c$'s left
		\STATE Let $c_{i+1}$ be the leftmost candidate in $S$ on $c$'s right 
		\STATE $\hat{C} \leftarrow \hat{C} \cup 
		        \big\{ \text{Distant-Candidate}(\C[c_i, c])\} \cup \{ \text{Distant-Candidate}(\C[c, c_{i+1}]) \big\}$
	\ENDIF
\ENDWHILE
\RETURN $S$
\end{algorithmic}
\end{algorithm}

\begin{proposition}\label{pr:distant-candidate}
Let $S = \{ c_1, \ldots, c_\ell \}$ be the set of currently elected candidates in Algorithm~\ref{alg:greedy}, and let $(\hat{c}_1, \delta_1), \ldots, (\hat{c}_{\ell-1}, \delta_{\ell-1})$ be the candidate-distance pairs maintained in $\hat{C}$, as defined in \eqref{eq:most-distant}. Then, 
\[ \max_{c \in \C} \{ d(c, S) \} = 
   \max_{i \in [\ell-1]} \big\{ d(\hat{c}_i, \{ c_i, c_{i+1} \}) \big\} =
   \max_{i \in [\ell-1]} \big\{ \delta_i \big\}  
   \]
\end{proposition}

\begin{proof}
Since $c_1$ (resp. $c_\ell$) is the leftmost (resp. rightmost) candidate in $\C$, every candidate in $\C \setminus S$ belongs to one of the intervals $\C[c_1, c_2], \ldots, \C[c_{\ell-1}, c_\ell]$. We let the furthest candidate $c$ to $S$ lie in the interval $\C[c_i, c_{i+1}]$. Then, $c$ has to be the candidate $\hat{c}_i \in \C[c_i, c_{i+1}]$ with maximum distance to the endpoints $\{ c_i, c_{i+1}\}$ and $d(\hat{c}_i, S) = d(\hat{c}_i, \{ c_i, c_{i+1} \}) =\delta_i$.
\qed\end{proof}


\begin{algorithm}[t]
\caption{\label{alg:dist-cand}The Distant-Candidate algorithm}
\textbf{Input}: Candidate interval $\C[c, c']$,
      a voter $v \in \Cl(c'')$ for every $c'' \in \C[c, c']$. \\
\textbf{Output}: Candidate $\hat{c} \in \C[c, c']$ with maximum $d(\hat{c}, \{ c, c'\})$

\smallskip
\begin{algorithmic}[1]
\IF{$\big|\C[c, c']\big| = 3$}
	\STATE  $c'' \leftarrow \C[c, c'] \setminus \{ c, c'\}$
	\RETURN $(c'', \min\{ d(c'', c), d(c'', c')\})$
\ENDIF
\STATE Let $c''$ be the leftmost candidate in $\C[c, c'] \setminus \{ c \}$ \COMMENT{henceforth, $\big|\C[c, c']\big| > 3$}
\WHILE{$c'' \in C[c, c']$}
	\STATE Let $\succ_{c''}$ be the ranking $\succ_v$ of any $v \in \Cl(c'')$ \COMMENT{true $\succ_{c''}$ is not known; the algorithm works as if $\succ_{c''}=\succ_v$}
	\IF{$c' \succ_{c''} c$}
		\STATE \label{line9ofalgo3} Let $c_r$ be $c''$ and $c_l$ be next candidate on $c''$'s left
		\COMMENT{$c_l$ and $c_r$ found, while-loop terminates}
		\STATE \textbf{break-while-loop} 
    \ELSE 
   		\STATE $c'' \leftarrow$\, the next candidate on $c''$'s right 
   		\COMMENT{proceed to the next candidate on the right}
   \ENDIF
\ENDWHILE

\IF{$d(c, c_l) \geq d(c_r, c')$}
	\RETURN $(c_l, \min\{d(c, c_l), d(c', c_l)\})$
\ELSE
	\RETURN $(c_r, \min\{d(c_r, c), d(c_r, c')\})$
\ENDIF
\end{algorithmic}
\end{algorithm}

\subsection{The Distant-Candidate Algorithm}

The Distant-Candidate algorithm, formally described in Algorithm~\ref{alg:dist-cand}, receives as input two candidates $c, c'$, makes at most $3$ distance queries, and returns a pair $(\hat{c}, \delta)$, where $\hat{c} \in \C[c, c']$ is the most distant candidate to $\C[c, c']$'s endpoints and $\delta = d(\hat{c}_i, \{ c, c' \})$. 

For the intuition, we initially assume access to the rankings $\succ_c$ and $\succ_{c'}$, where all candidates in $\C$ are listed in increasing order of distance to $c$ and $c'$, respectively. 
Due to the $1$-dimensional structure of $\C$, for every $c, c' \in \C$, with $c < c'$ and $|\C[c, c']| \geq 4$, the most distant candidate $\hat{c} = \arg\max_{c'' \in \C[c, c']} \{ d(c'', \{c, c'\}) \}$ can be computed as follows:
Starting with $c$ and moving from left to right on the candidate interval $\C[c, c']$, we find the rightmost candidate $c_l \in \C[c, c']$ that prefers $c$ to $c'$ and the leftmost candidate $c_r \in \C[c, c']$ that prefers $c'$ to $c$. We note that $c_l$ and $c_r$ can be found using only ordinal information, that they are next to each other on the candidate axis, and that $\hat{c}$ must be either $c_l$ or $c_r$. Then, $d(c_l, \{c, c'\}) = d(c_l, c)$ and $d(c_r, \{c, c'\}) = d(c_r, c')$. Hence, $\hat{c}$ is $c_l$, if $d(c_l, c) > d(c_r, c')$, and $c_r$ otherwise, which can be determined by $2$ distance queries $d(c_l, c)$ and $d(c_r, c')$. 

The implementation of Distant-Candidate in Algorithm~\ref{alg:dist-cand}, computes $(\hat{c}, d(\hat{c}, \{c, c'\})$ based on the ordinal information submitted by the voters. More precisely, Algorithm~\ref{alg:dist-cand} uses a ranking $\succ_v$, submitted by a voter $v \in \Cl(c)$, instead of the ranking $\succ_c$ above (to which the algorithm does not have access). Since $\succ_v$ and $\succ_c$ may differ (because the locations of $v$ and $c$ may be different), the algorithm needs to use some additional information provided by a 3rd distance query. The correctness of the Distant-Candidate algorithm is formally established by the following: 


\begin{lemma}\label{l:most-distant}
For any $c, c' \in \C$, with $c < c'$ and $|\C[c, c']| \geq 3$, Algorithm~\ref{alg:dist-cand} correctly returns the candidate $\hat{c} \in \C[c, c']$ with maximum distance to the interval's endpoints $\{c, c'\}$, i.e., the candidate
\[ \hat{c} = \arg\max_{c'' \in \C[c, c']}\big\{ d(c'', \{c, c'\}) \big\}\,, \]
along with its distance $d(\hat{c}, \{ c, c'\}) = d(\hat{c}, S)$ to $S$.
\end{lemma}

\begin{proof}
If $|\C[c, c']| = 3$, then there is a single $c'' \in \C[c, c'] \setminus \{c, c'\}$ that is necessarily the most distant candidate. Algorithm~\ref{alg:dist-cand} returns $c''$ and its distance $d(c'', \{c, c'\})$ (computed using $2$ distance queries). 

To give the intuition for the most interesting case where $|\C[c, c']| \geq 4$, we first consider candidate-restricted instances, where for any candidate $c'' \in \C[c, c']$, all voters $v \in \Cl(c'')$ are collocated with $c''$. Hence, the ranking $\succ_v$ submitted by some voter $v \in \Cl(c'')$ and used in step~7 of Algorithm~\ref{alg:dist-cand} is identical to the true ranking $\succ_{c''}$, where all candidates in $\C$ appear in increasing order of their distance to $c''$. 

We consider the midpoint $\mu = (c+c')/2 \in \reals$, which is the point in the real interval $[c, c']$ with maximum distance $d(\mu, \{c, c'\})$ to the endpoints $\{ c, c'\}$. In Algorithm~\ref{alg:dist-cand}, $c_r$ is the leftmost candidate in $\C[c, c']$ that is closer to the right endpoint $c'$ than to the left endpoint $c$ (and thus  $d(c_r, \{c, c'\}) = d(c_r, c')$). By the definition of $c_r$, $c_l$ is the rightmost candidate in $\C[c, c']$ that is closer to the left endpoint $c$ than to the right endpoint $c'$ (and thus $d(c_l, \{c, c'\}) = d(c_l, c)$). Therefore, $c_l \leq \mu \leq c_r$, with at least one inequality strict, and with no other candidate between $c_l$ and $c_r$. Hence, $c_l$ and $c_r$ are the candidates in $\C[c, c']$ closest to $\mu$. Then the maximum of $d(c_l, c) = d(c_l, S)$ and $d(c_r, c') = d(c_r, S)$ determines the candidate in $\C[c, c']$ with maximum distance to its endpoints. 

We next remove the assumption that Algorithm~\ref{alg:dist-cand} has access to the true ranking $\succ_{c''}$ but use, as a substitute, the ranking $\succ_v$ of any $v \in \Cl(c'')$. Let $\hat{c}$ be the candidate in $\C[c, c']$ with maximum distance to the endpoints $\{ c, c'\}$. $\hat{c}$ is the closest candidate to the midpoint $\mu$. We show that $\hat{c} \in \{ c_l, c_r \}$ (cf. line \ref{line9ofalgo3} of Algorithm~\ref{alg:dist-cand} for the definition of $\{ c_l, c_r \}$) and that Algorithm~\ref{alg:dist-cand} correctly returns $\hat{c}$. 

Without loss of generality, we assume that $\hat{c} \leq \mu$ (the case where $\hat{c} > \mu$ is symmetric). 
Let $\hat{v} \in \Cl(\hat{c})$ be the voter whose preference list $\succ_{\hat{v}}$ is used in place of $\succ_{\hat{c}}$ in Algorithm~\ref{alg:dist-cand} ($\hat{v}$ can be any voter in $\Cl(\hat{c})$). Moreover, let $c_a$ (resp. $c_b$) be the next candidate on the left (resp. right) of $\hat{c}$ in $\C[c, c']$ (we note that $c_a$ can be $c$ and $c_b$ can be $c'$, but not both). Let $v_a$ (resp. $v_b$) be the voter in $\Cl(c_a)$ (resp. $\Cl(c_b)$) whose preference list $\succ_{v_a}$ (resp. $\succ_{v_b}$) is used in place of $\succ_{c_a}$ (resp. $\succ_{c_b}$). 

Since $\hat{c} \leq \mu$, we have that $c \succ_{v_a} c'$, because $c_a < \hat{c} \leq \mu$ and $d(v_a, c_a) < d(v_a, \hat{c})$. Therefore, $v_a < \hat{c} \leq \mu$. Moreover, we have that $c' \succ_{v_b} c$, which holds because $\hat{c} \leq \mu < c_b$, $\hat{c}$ is the closest candidate to $\mu$ and $d(v_b, c_b) < d(v_b, \hat{c})$ (also recall that we do not allow any ties). Hence $\mu < v_b$ (due to the fact that $\mu$ is closer to $\hat{c}$; if $v_b \leq \mu$, it would be $d(\hat{c}, v_b) < d(c_b, v_b)$ and $v_b$ would be in $\Cl(\hat{c})$). 

Because $c \succ_{v_a} c'$ (and thus $c_a$ is not $c_r$) and $c' \succ_{v_b} c$ (and thus $c_r$ is either $\hat{c}$ or $c_b$), $\hat{c}$ is either $c_l$ or $c_r$. We distinguish two cases depending on the placement of $\hat{v}$ with respect to $\mu$. 

If $\hat{v} \leq \mu$, $c \succ_{\hat{v}} c'$ (recall that we do not allow any ties in the rankings profile), which implies that $c_r = c_b$ and $c_l = \hat{c}$. Moreover, $d(\hat{c}, c) \geq d(c_b, c')$, because $\hat{c}$ is the furthest candidate to $\{c, c'\}$. Therefore, Algorithm~\ref{alg:dist-cand} returns $\hat{c}$ as the furthest candidate and the distance $d(\hat{c}, c) = d(\hat{c}, S)$. 

If $\mu < \hat{v}$ and $c' \succ_{\hat{v}} c$, then  $c_r = \hat{c}$ and $c_l = c_a$. Furthermore, $d(\hat{c}, c') > d(c_a, c)$, due to the fact that $c \leq c_a < \hat{c} \leq \mu < c'$. Therefore, Algorithm~\ref{alg:dist-cand} returns $\hat{c}$ as the furthest candidate and the distance $\min\{ d(\hat{c}, c), d(\hat{c}, c') \} = d(\hat{c}, S)$. 
\qed\end{proof}

\subsection{Putting Everything Together: The Distortion of Algorithm~\ref{alg:greedy}}

The following formally establishes the distortion achieved by the query-efficient implementation of the greedy algorithm for $k$-center (Algorithm~\ref{alg:k-center}) presented in Algorithm~\ref{alg:greedy}.

\begin{theorem}\label{thm:k-queries}
For any $k \geq 3$, Algorithm~\ref{alg:greedy} achieves a distortion of at most $5n$ for the social cost (and at most $5$ for the egalitarian cost) for $k$-Committee Election with $1$-Euclidean preferences using at most $6k-15$ candidate distance queries. 
\end{theorem}

\begin{proof}
Algorithm~\ref{alg:greedy} adapts Algorithm~\ref{alg:k-center} to unknown candidate locations. It calls Distant-Candidate, which  computes the most distant candidate $\hat{c}_i$ in each interval $\C[c_i, c_{i+1}]$ defined by the candidates already elected in $S$, using at most three candidate distance queries. Proposition~\ref{pr:distant-candidate} shows that the candidate $\hat{c}_i$ with maximum distance $\delta_i$ among them is the most distant candidate to $S$, which is added to $S$ in step~5. 

The Distant-Candidate algorithm is called once in step~2 and $2(k-3)$ times in step~10 (twice in each while-loop iteration, for $|S| = 3, \ldots, k-1$). So, the total number of candidate distance queries is at most $6(k-3)+3$. The correctness of Algorithm~\ref{alg:greedy} (i.e., the fact that in each iteration, the candidate $c$ with maximum $d(c, S)$ is added to $S$) follows from Lemma~\ref{l:most-distant} and Proposition~\ref{pr:distant-candidate}. The distortion bound for the egalitarian cost uses that Algorithm~\ref{alg:k-center} is $2$-approximate for the egalitarian cost in candidate-restricted instances \cite[Theorem~2.3]{WS10}. Then,  Theorem~\ref{thm:factor3_egal}, in Suppl.~\ref{s:app:good-egalitarian}, implies an upper bound of $5$ on the distortion for the egalitarian cost in the original instances. The bound of $5n$ on the distortion for the social cost holds because for any $S \subseteq \C$, $\Egal(S) \leq \Util(S) \leq n\,\Egal(S)$. 
\qed\end{proof}

\begin{remark}
We next show that the distortion of Algorithm~\ref{alg:greedy} (in fact, the distortion of the voter clustering computed by Algorithm~\ref{alg:greedy}) is $\Omega(n)$ for the social cost, for any $k \geq 3$. 

We consider an instance with a committee of size $k = 3$, $m = 5$ candidates, $a$ located at $0$, $b$ located at $1$, $c$ located at $2+\eps$, $d$ located at $4+3\eps$ and $e$ located at $8+4\eps$, where $\eps > 0$ is small and used for tie-breaking, and $n$ voters, the leftmost voter is collocated with candidate $a$, $n/2-1$ voters are collocated with candidate $b$, $n/2-2$ other voters are collocated with candidate $c$, another voter is collocated with candidate $d$ and the rightmost voter is collocated with candidate $e$. 

The optimal $3$-committee is $\{ b, c, e \}$ with social cost $3+2\eps \approx 3$. The $3$-committee elected by Algorithm~\ref{alg:greedy} is $\{ a, d, e \}$ with social cost $3n/2-5 + (n/2-2)\eps \approx 3n/2$ and distortion that tends to $n/2$ as $n$ tends to infinity and $\eps$ tends to $0$. 


In an attempt to improve the distortion, it is reasonable to try the following: after Algorithm~\ref{alg:greedy} has elected its committee, which imposes a clustering on the set of voters, we consider the corresponding clusters and elect the median candidate of each cluster. In our example, we consider the voters collocated with $a$, $b$ and $c$, who are all assigned to candidate $a$ in the committee elected by Algorithm~\ref{alg:greedy}, and replace candidate $a$ with their median candidate $b$. Then, we obtain the committee $\{ b, d, e \}$ with social cost $1+(n/2-2)(1+\eps) \approx n/2$ and distortion that tends to $n/6$ as $n$ tends to infinity and $\eps$ tends to $0$. 
\end{remark}

\section{Low Distortion via Good Candidate Subsets}
\label{s:good}

We next analyze the distortion achieved by the optimal $k$-committee of the candidate-restricted instance induced by a small representative set of candidates. 

We say that a candidate subset $\C' \subseteq \C$ is \textbf{$(\ell, \beta)$-good}, for some  $\ell \geq k$ and some $\beta \geq 1$, if $|\C'| = \ell$ and 
$\Util(\C') \leq \beta\,\Util(S^\ast)$, where $S^\ast$ is an optimal $k$-committee for the original instance. 
%
%
Namely, $\C'$ is \emph{$\ell$-sparse}, in the sense that $\C'$ includes $\ell \leq m$ candidates (ideally $\ell \ll m$), and is \emph{$\beta$-good}, in the sense that representing each voter by her top candidate in $\C'$ imposes a social cost at most $\beta$ times the optimal social cost. The original set $\C$ of candidates is $(m, 1)$-good, while any $k$-committee with distortion $\beta$ is $(k, \beta)$-good. 

Given an $(\ell, \beta)$-good set of candidates $\C'$, we let $\C'_{\cres} = \{ (c_1, n_1), \ldots, (c_\ell, n_\ell) \}$ denote the candidate-restricted instance induced by $\C'$. Then, $c_1 < \cdots < c_\ell$ denote the locations of candidates in $\C'$ on the line, and $n_i = |\Cl(c_i)|$  is the number of voters with $c_i$ as their top candidate in $\C'$. We maintain that $n_1+\cdots+n_\ell = n$ and that each $n_i > 0$ (the latter by removing inactive candidates from $\C'$). 

The following shows that an optimal $k$-committee for the candidate-restricted instance $\C'_{\cres}$ induced by an $(\ell, \beta)$-good set $\C'$ achieves a distortion of $1+2\beta$ for the original instance.  

\begin{theorem}\label{thm:factor3}
Let $(\C, \V)$ be an instance of the $k$-Committee Election, let $\C' \subseteq \C$ be an $(\ell, \beta)$-good set, let $\C'_{\cres}$ be the candidate-restricted instance induced by $\C'$ and let $S$ (resp. $S^\ast$) be an optimal $k$-committee for $\C'_{\cres}$ (resp. for $(\C, \V)$). Then, $\Util(S) \leq (1+2\beta)\Util(S^\ast)$.
\end{theorem}

\begin{proof}
For each voter $v$ (with her location $v$ as in the original instance), we let $\top'(v) \in \C'$ be $v$'s top candidate in $\C'$. Then, by the triangle inequality, $d(v, S) \leq d(v, \top'(v)) + d(\top'(v), S)$. Summing up over all voters $v \in \V$, we obtain:
\begin{equation}\label{eq:beta_good}
  \Util(S) \leq \Util(\C') + \Util(\C'_{\cres}, S)\,, 
\end{equation} 
where $\Util(\C'_{\cres}, S) = \sum_{v \in V} d(\top'(v), S) = \sum_{i=1}^\ell n_i \,d(c_i, S)$ is the social cost of $S$ for the candidate-restricted instance $\C'_{\cres}$ induced by $\C'$, and $\Util(\C') = \sum_{v\in \V} d(v, \top'(v)) = \sum_{v\in \V} d(v, \C')$. 

We observe that $\Util(\C'_{\cres}, S) \leq \Util(\C'_{\cres}, S^\ast)$, because, as shown in Proposition~\ref{pr:active}, Suppl.~\ref{s:app:active-candidates} (and since voters are collocated with their top candidate in $\C'_{\cres}$), in the candidate-restricted instance $\C'_{\cres}$, we can replace candidates in $S^\ast \setminus \C'$ with candidates in $S$ without increasing the social cost. Moreover, since $d(\top'(v), S^\ast) \leq d(\top'(v), v) + d(v, S^\ast)$, 
we obtain that $\Util(\C'_{\cres}, S^\ast) \leq \Util(\C')+\Util(S^\ast)$. 

Combined with the observations above, \eqref{eq:beta_good} implies that:
\[ \Util(S) \leq 2\Util(\C') + \Util(S^\ast) \leq (1+2\beta)\Util(S^\ast)\,,\]
where the second inequality holds because $\C'$ is a $(\ell, \beta)$-good set of candidates.
\qed\end{proof} 

As noted in Suppl.~\ref{s:app:axis}, as soon as we have the distances between all active candidates in an $(\ell, \beta)$-good set $\C'$, which requires $\ell-1$ distance queries, an optimal $k$-committee for the candidate-restricted instance $\C'_{\cres}$ induced by $\C'$ can be computed in polynomial time by dynamic programming.

\section{Hierarchical Partitioning for Good Candidate Subsets}
\label{s:coresets}

Next, we use hierarchical partitioning of the candidate axis (see Algorithm~\ref{alg:coreset}) and compute a $(O(k\log n), O(1))$-good set of candidates with $O(k\log n)$ distance queries. 

\begin{algorithm}[t]
\caption{\label{alg:coreset}Hierarchical partitioning of $\C[c_1, c_m]$}
\textbf{Input}: Candidates $\C = \{ c_1, \ldots, c_m \}$, $k \in \{2, \ldots, m-1\}$, voter ranking profile $\vec{\succ} = (\succ_1, \ldots, \succ_n)$\\
\textbf{Output}: Partitioning $\I$ of $\C$ into $O(k\log n)$ intervals. 

\begin{algorithmic}[1]
\STATE Let $S = \{c^1, \ldots, c^k\}$ be the result of Algorithm~\ref{alg:greedy}

\STATE $\I \leftarrow \{\, (\C[c_a^1, c_b^1], n_1, d(c_a^1, c_b^1)), \ldots, (\C[c_a^k, c_b^k], n_k, d(c_a^k, c_b^k))\,\}$ 
\COMMENT{Start with the partitioning of $(\C, \V)$ induced by $S$}
\STATE $\delta^\ast \leftarrow \max_{i\in [k]} \{ d(c^i_a, c^i_b) \}$
\WHILE{$|\I| \leq 7k (\log_2(5nk) + 2)$}
	\STATE Let $(\C[c_a, c_b], n_{ab}, d(c_a, c_b)) \in \I$ with $|\C[c_a, c_b]| \geq 4$ and maximum weight $\weight(c_a, c_b) = n_{ab} d(c_a, c_b)$ \\
%
%
    \STATE \textbf{if} \,$\weight(c_a, c_b) \leq \delta^\ast / (5k)$\, \textbf{then}\, \textbf{break-while-loop} 
	\STATE $\I \leftarrow \big(\I \setminus \big\{ (\C[c_a, c_b], n_{ab}, d(c_a, c_b)) \big\}\big) \cup \mathrm{Partitioning}(\C[c_a, c_b])$
\ENDWHILE
\RETURN $\I$
\end{algorithmic}
\end{algorithm}

In Algorithm~\ref{alg:coreset}, each triple (or \emph{interval}) in $\I$ consists of a candidate interval $\C[c_a, c_b]$, the number $n_{ab}$ of voters $v$ with $\top(v) \in \C[c_a, c_b]$, and the interval length $d(c_a, c_b)$. We refer to $\weight(c_a, c_b) = n_{ab}d(c_a, c_b)$ as the \emph{weight} of interval $\C[c_a, c_b]$. Algorithm~\ref{alg:coreset} computes a partitioning $\I$ of the candidate axis $\C[c_1, c_m]$ into at most $7k (\log_2 (5nk) + 2)$ intervals by repeatedly splitting the interval in $\I$ with largest weight (and at least $4$ candidates) into two intervals.

Algorithm~\ref{alg:coreset} starts with the partitioning induced by the $k$-committee $S = \{ c^1, \ldots, c^k\}$ computed by Algorithm~\ref{alg:greedy} (we let $c^1 < c^2 < \cdots < c^k$). In step~2, for each $i \in [k]$, $\C[c^i_a, c^i_b]$ is the interval that includes all candidates in $\C$ closer to $c^i$ than to any other candidate in $S$ ($c_a^i$, resp. $c_b^i$, is the leftmost, resp. the rightmost, such candidate), $n_i$ is the number of voters $v$ with $\top(v) \in \C[c^i_a, c^i_b]$, and $d(c_a^i, c_b^i)$ is the length of $\C[c^i_a, c^i_b]$. We let $\delta^\ast = \max_{i\in [k]} \{ d(c^i_a, c^i_b) \}$. Then, $\Util(S^\ast) \geq \delta^\ast / 5$, where $S^\ast$ is an optimal $k$-committee for the original instance, because $S$ has distortion $5$ for the egalitarian cost and we assume that all candidates in $\C$ are active. 
%

For the interval split, in step~4, we use the Partitioning algorithm (Algorithm~\ref{alg:partitioning}), which is a modification of the algorithm Distant-Candidate used in Section~\ref{s:greedy}. The Partitioning algorithm uses at most $4$ distance queries and splits each interval $\C[c_a, c_b]$ into two intervals $(\C[c_a, c_l], n_{al}, d(c_a, c_l))$ and $(\C[c_r, c_b], n_{rb}, d(c_r, c_b))$, where $c_l$ (resp. $c_r$) is the rightmost (resp. leftmost) candidate in $\C[c_a, c_b]$ on the left (resp. right) of the interval's midpoint $(c_a+c_b)/2$. We defer the detailed description and the analysis of the Partitioning algorithm to Section~\ref{sec:partitioning}, at the end of this section. 

We note that $\C[c_a, c_l] \cup \C[c_r, c_b] = \C[c_a, c_b]$. Therefore, since initially, in step~2, $\I$ is a partitioning of $(\C, \V)$, the collection of intervals $\I$ maintained by Algorithm~\ref{alg:coreset} always remains a partitioning of $\C$. 
To obtain an $(O(k\log n), O(1))$-good set $\C'(\I) \subseteq \C$ from $\I$, we include in $\C'(\I)$ the endpoints $c_a$ and $c_b$ of (resp. all candidates in) each interval $\C[c_a, c_b]$ in $\I$ with more than (resp. at most) $3$ candidates. Since $|\I| = O(k\log n)$, $\C'(\I)$ consists of $O(k\log n)$ candidates. The following shows that $\Util(\C'(\I)) \leq 2\,\Util(S^\ast)$, where $S^\ast$ is an optimal $k$-committee for the original instance. 

\begin{theorem}\label{thm:coreset}
Let $\I$ be the partitioning of $\C$ computed by Algorithm~\ref{alg:coreset}. Then, the resulting set $\C'(\I) \subseteq \C$ is a $(O(k\log n), 2)$-good set of candidates.
\end{theorem}

\begin{proof}
We first upper bound $\Util(\C'(\I))$. We call an interval $\C[c_a, c_b] \in \I$ \emph{expensive}, if $\C[c_a, c_b] \cap S^\ast \neq \emptyset$ (i.e., $\C[c_a, c_b]$ includes an optimal candidate), and \emph{cheap} otherwise. 
For each voter $v$ with $\top(v)$ in a cheap interval $\C[c_a, c_b]$, $\cost_v(\C') \leq \cost_v(S^\ast)$, because the interval's endpoints $c_a, c_b \in \C'$ and $v$'s nearest candidate in $S^\ast$ is outside $\C[c_a, c_b]$. Therefore, the total contribution to $\Util(\C'(\I))$ of all voters associated with cheap intervals in $\I$ is at most their contribution to $\Util(S^\ast)$. 
The contribution to $\Util(\C'(\I))$ of the voters $v$ with $\top(v)$ in an expensive interval $\C[c_a, c_b]$ is at most their contribution to $\Util(S^\ast)$ plus $\weight(c_a, c_b)$. Moreover, since each expensive interval includes a candidate of $S^\ast$, there are at most $k$ expensive intervals in $\I$. 
We next show, by adapting an argument of \cite{FrahlingS05},
that as soon as $|\I| > 7k (\log_2(5nk) + 1)$, each interval in $\I$ has weight at most $\Util(S^\ast) /k$. 
Therefore, at this point, the additional cost due to the total weight of expensive intervals is at most $\Util(S^\ast)$. 

We refer to an interval $\C[c_a, c_b]$ as \emph{light}, if $\weight(c_a, c_b) \leq \Util(S^\ast) / k$, and as \emph{heavy}, otherwise. We observe that splitting an interval $\C[c_a, c_b]$ of weight $\weight(c_a, c_b)$ results in two intervals $\C[c_a, c_l]$ and $\C[c_r, c_b]$ each of length (resp. weight) at most $d(c_a, c_b)/2$ (resp. $\weight(c_a, c_b)/2$). 
Therefore, splitting a light interval replaces it with two light intervals in $\I$. Moreover, since Algorithm~\ref{alg:coreset} always splits the heaviest interval in $\I$, from the first iteration that Algorithm~\ref{alg:coreset} splits a light interval and on, all intervals in $\I$ are light. 

A heavy interval $\C[c_a, c_b]$ is \emph{close}, if $d(\{c_a,  c_b\}, S^\ast) < d(c_a, c_b)$, and is \emph{far} otherwise. 
We say that an interval $\C[c_a, c_b]$ is \emph{encountered} by Algorithm~\ref{alg:coreset}, if there is a point during its execution where $\C[c_a, c_b] \in \I$. 
%
%
We next show that (i) the total number of far heavy intervals encountered by Algorithm~\ref{alg:coreset} is at most $2k(\log_2(5nk) + 1)$\,; and (ii) that the total number of close heavy intervals encountered by Algorithm~\ref{alg:coreset} is at most $5k (\log_2(5nk) + 1)$.

To upper bound the number of heavy intervals, we partition the intervals encountered by Algorithm~\ref{alg:coreset} in levels according to their length. An interval $\C[c_a, c_b]$ is \emph{level-$i$}, if $2^{i-1} \delta^\ast < d(c_a, c_b) \leq 2^{i} \delta^\ast$. Algorithm~\ref{alg:coreset} starts with intervals at a level $i \leq 0$, because initially all intervals created in step~2 have $d(c^i_a, c^i_b) \leq \delta^\ast$. Algorithm~\ref{alg:coreset} can encounter intervals until levels down to $i = -\log_2 (5nk)$. A level-$(-\log_2 (5nk))$ interval $\C[c_a, c_b]$ has length $d(c_a, c_b) \leq \delta^\ast / (5nk)$ and weight $\weight(c_a, c_b) \leq n d(c_a, c_b) \leq \delta^\ast / (5k)$. Since $\Util(S^\ast) \geq \delta^\ast/5$, if the weight of the heaviest interval in $\I$ is at most $\delta^\ast / (5k) \leq \Util(S^\ast) / k$, all intervals in $\I$ are light and Algorithm~\ref{alg:coreset} terminates through step~6. We note that we treat intervals with at most $3$ candidates as light, because all candidates in such intervals are included in $\C'(\I)$. 

Thus, level-$(-\log_2 (5nk))$ intervals never split and stay in $\I$. Moreover, the set of level-$i$ intervals encountered by Algorithm~\ref{alg:coreset}, for any $i = -\log_2 (5nk)), \ldots, 0$, 
form a partitioning of a subset of $\C$ 
(assuming that this set is non-empty). 

The voters associated with a far heavy interval $\C[c_a, c_b]$ contribute to $\Util(S^\ast)$ a total cost no less than $\weight(c_a, c_b)/2$. This holds because for any voter $v$ with $\top(v)$ in a far heavy interval $\C[c_a, c_b]$, $d(c_a, c_b)$, which is $v$'s contribution to $\weight(c_a, c_b)$, is at most $2d(v, S^\ast)$. 

Specifically, let $c^\ast_v \in S^\ast$ be $v$'s closest candidate in $S^\ast$. If $c^\ast_v < c_a < v$ (or symmetrically, $v > c_b > c^\ast_v$), then
\[ d(S^\ast, v) = d(c^\ast_v, v) \geq d(S^\ast, \{c_a,  c_b\}) \geq d(c_a, c_b) \,,\]
due to definition of far heavy intervals. 
If $c^\ast_v < v < c_a$ (or symmetrically, $c_b > v > c^\ast_v$), then $d(c^\ast_v, v) \geq d(v, \{ c_a, c_b \})$, because $\top(v) \in \C[c_a, c_b]$. Hence, 
\[ d(c^\ast_v, \{c_a, c_b\}) = d(c^\ast_v, v) + d(v, \{c_a, c_b\}) \geq 2d(c^\ast_v, v)\,. \]
Therefore, 
\begin{align*} 
  d(S^\ast, v) & = d(c^\ast_v, v) \geq d(c^\ast_v, \{c_a, c_b\})/2 \\
  & \geq d(S^\ast, \{c_a,  c_b\}) /2 \geq d(c_a, c_b)/2 \,,
\end{align*}
where the last inequality follows from the definition of far heavy intervals.

The far heavy level-$i$ intervals encountered by Algorithm~\ref{alg:coreset}, at any fixed level $i$, induce a partitioning of a subset of voters. Since each far heavy interval has $\weight(c_a, c_b) \geq \Util(S^\ast) / k$, the total number of far heavy intervals encountered by Algorithm~\ref{alg:coreset} at any fixed level-$i$ is at most $2k$, and at most $2k(\log_2(5nk)+1)$ in total, which concludes the proof of (i) above. 

Furthermore, each candidate in the optimal $k$-committee $S^\ast$ can be associated with at most $5$ close heavy level-$i$ intervals, for any level $i$. This holds due to the definition of close heavy intervals as $d(\{c_a, c_b\}, S^\ast) < d(c_a, c_b)$, the fact that the lengths of level-$i$ intervals are within a factor of at most $2$ from each other, and the fact that the close heavy level-$i$ intervals encountered by Algorithm~\ref{alg:coreset}, at any fixed level $i$, induce a partitioning of a subset of $\C$. Therefore, the total number of close heavy intervals encountered by Algorithm~\ref{alg:coreset} is at most $5k(\log_2(5nk)+1)$, which concludes the proof of (ii) above.  

Algorithm~\ref{alg:coreset} keeps splitting heavy intervals, as long as they exist in $\I$. Light intervals created by such splits are accumulated in $\I$ and are not split, as long as heavy intervals exist in $\I$. Since the total number of heavy intervals encountered by Algorithm~\ref{alg:coreset} is at most $7k(\log_2 (5nk)+1)$, the first split of a light interval happens no later than iteration $7k(\log_2 (5nk)+1)+1$. At that point all intervals in $\I$ are light. 
%
%
Moreover, since we start with $|\I| = k$, the total number of intervals in $\I$ at that point is at most $7k(\log_2 (5nk)+1)+k+1 < 7k(\log_2 (5nk)+2) = O(k\log n)$ (since $n \geq k$). 
%
%
%
\qed\end{proof}


Algorithm~\ref{alg:coreset} performs $O(k\log n)$ splits during its execution. Since the Partitioning algorithm uses at most $4$ distance queries per interval split, the total number of distance queries used by Algorithm~\ref{alg:coreset} is at most $O(k\log n)$. Combining Theorem~\ref{thm:factor3} (and the discussion below it) with the analysis of Algorithm~\ref{alg:coreset} in Theorem~\ref{thm:coreset}, we obtain that:


\begin{theorem}\label{thm:constant-distortion}
There is a polynomial-time deterministic rule for $k$-Committee Election that uses $O(k\log n)$ distance queries and achieves a distortion of at most $5$. 
\end{theorem}

\subsection{Description and Analysis of the Partitioning Algorithm}
\label{sec:partitioning}

\begin{algorithm}[!b]
\caption{\label{alg:partitioning}The Partitioning algorithm}
\textbf{Input}: Candidate interval $\C[c, c']$, 
rankings $\succ_v$ for all voters $v \in \bigcup_{c'' \in \C[c, c']}\Cl(c'')$\\ 
\textbf{Output}: Intervals $(\C[c, c_l], n_l, d(c, c_l))$ and $(\C[c_r, c'], n_r, d(c_r, c'))$ subdividing interval $(\C[c, c'], n, d(c, c'))$

\smallskip
\begin{algorithmic}[1]
\STATE Let $c''$ be the leftmost candidate in $\C[c, c'] \setminus \{ c \}$
\WHILE{$c'' \in C[c, c']$}
	\STATE Let $\succ_{c''}$ be the ranking $\succ_v$ of any $v \in \Cl(c'')$ 
	\IF{$c' \succ_{c''} c$}
		\STATE Let $c_r$ be $c''$ and $c_l$ be next candidate on $c''$'s left 
		\COMMENT{$c_l$ and $c_r$ found, while-loop terminates}
		\STATE \textbf{break-while-loop} 
    \ELSE 
   		\STATE $c'' \leftarrow$\, the next candidate on $c''$'s right 
   		\COMMENT{proceed to the next candidate on the right}
   \ENDIF
\ENDWHILE

\IF{$d(c, c_l) \geq d(c_r, c')$}
	\STATE \COMMENT{$c_l$ is the most distant candidate to $\{c, c'\}$}
	\IF{$d(c, c_l) > d(c_l, c')$}
		\STATE $c_r \leftarrow c_l$ 
		\COMMENT{$c_l$ is the first candidate on the right of $(c+c')/2$}
		\STATE Let $c_l$ be the first candidate on $c_r$'s left
	\ENDIF
\ELSE 
	\STATE \COMMENT{$c_r$ is the most distant candidate to $\{c, c'\}$}
	\IF{$d(c, c_r) < d(c_r, c')$}
	\STATE $c_l \leftarrow c_r$ 
		\COMMENT{$c_r$ is the first candidate on the left of $(c+c')/2$}
		\STATE Let $c_r$ be the first candidate on $c_l$'s right
	\ENDIF
\ENDIF
\STATE $n_l \leftarrow \sum_{\tilde{c} \in \C[c, c_l]} |\Cl(\tilde{c})|$ 
\COMMENT{$c_l$ is the first candidate on the left of $(c+c')/2$}
\STATE $n_r \leftarrow \sum_{\tilde{c} \in \C[c_r, c']} |\Cl(\tilde{c})|$ 
\COMMENT{$c_r$ is the first candidate on the right of $(c+c')/2$}
\RETURN $\big\{ (\C[c, c_l], n_l, d(c, c_l)), (\C[c_r, c'], n_r, d(c_r, c')) \big\}$
\end{algorithmic}
\end{algorithm}

We conclude this section by describing the Partitioning algorithm (Algorithm~\ref{alg:partitioning}) and verifying its main properties used by Algorithm~\ref{alg:coreset}. Lemma~\ref{l:partitioning} below shows that for any candidate interval $\C[c, c']$ with $|\C[c, c']| \geq 4$, Algorithm~\ref{alg:partitioning} correctly computes the candidate $c_l$, which is the rightmost candidate on the left of the midpoint $(c+c')/2$, and the candidate $c_r$, which is the leftmost candidate on the right of the midpoint $(c+c')/2$. Therefore, the two intervals $\C[c, c_l]$ and $\C[c_r, c']$ form a partitioning of the input interval $\C[c, c']$ and have length $\max\{d(c, c_l), d(c_r, c')\} \leq d(c, c')/2$. 

\begin{lemma}\label{l:partitioning}
For any $c, c' \in \C$, with $c < c'$ and $|\C[c, c']| \geq 4$. Algorithm~\ref{alg:partitioning}  computes the rightmost candidate $c_l$ on the left of the midpoint $(c+c')/2$ and the leftmost candidate $c_r$ on the right of the midpoint $(c+c')/2$ of interval $\C[c, c']$. 
\end{lemma}

\begin{proof}
We note that if $|\C[c, c']| \geq 4$, the first ten steps of Algorithm~\ref{alg:partitioning} are identical to the first ten steps of Algorithm~\ref{alg:dist-cand} (i.e., steps 5 to 14, applied to this case). Therefore, by the proof of Lemma~\ref{l:most-distant}, when Algorithm~\ref{alg:partitioning} reaches step~10, either $c_l$ or $c_r$ is the candidate $\hat{c} \in \C[c, c']$ with largest distance to $\{c, c'\}$. Then, by the proof of Lemma~\ref{l:most-distant}, if $d(c, c_l) \geq d(c_r, c')$, $\hat{c}$ is $c_l$, otherwise, $\hat{c}$ is $c_l$. In both cases, $\hat{c}$ is the candidate in $\C[c,c']$ closest to the midpoint $\mu = (c+c')/2$. In each case (i.e., either if $\hat{c} = c_l$, where steps 13 - 16 are executed, or if $\hat{c} = c_r$, where steps 19 - 22 are executed), Algorithm~\ref{alg:partitioning} distinguishes two subcases depending on whether $\hat{c}$ is on the left or on the right of $\mu$. 

In case where $d(c, c_l) \geq d(c_r, c')$ and $\hat{c} = c_l$, if $d(c, c_l) > d(c_l, c')$, $c_l$ is on the right of the midpoint $\mu$. Then, $c_l$ is in fact $c_r$ (i.e., the leftmost candidate on the right of $\mu$; so the value of the algorithm's variable $c_r$ is set to $c_l$ in step~14), and $c_l$ (i.e., the rightmost candidate on the left of $\mu$) is the first candidate on the left of $c_r$ on the candidate axis (step~15). Otherwise (i.e., if $d(c, c_l) \leq d(c_l, c')$), since $\hat{c} = c_l$ and $c_l$ and $c_r$ are consecutive on the candidate axis, $c_l$ is indeed the rightmost candidate on the left of $\mu$ and $c_r$ is the leftmost candidate on the right of $\mu$ (so the values of the corresponding algorithm's variables are set correctly). 

In case where $d(c, c_l) < d(c_r, c')$ and $\hat{c} = c_r$, if $d(c, c_r) < d(c_r, c')$, $c_r$ is on the left of the midpoint $\mu$. Then, $c_r$ is in fact $c_l$ (i.e., the rightmost candidate on the left of $\mu$; so the value of the algorithm's variable $c_l$ is set to $c_r$ in step~20), and $c_r$ (i.e., the leftmost candidate on the right of $\mu$) is the first candidate on the right of $c_l$ on the candidate axis (step~15). Otherwise (i.e., if $d(c, c_r) \geq d(c_r, c')$), since $\hat{c} = c_r$ and $c_l$ and $c_r$ are consecutive on the candidate axis, $c_r$ is indeed the leftmost candidate on the right of $\mu$ and $c_l$ is the rightmost candidate on the left of $\mu$ (so the values of the corresponding algorithm's variables are set correctly). 

Therefore, when Algorithm~\ref{alg:partitioning} reaches step~23, the value of the variable $c_l$ corresponds to the rightmost candidate on the left of the midpoint $(c+c')/2$ and the value of the variable $c_r$ corresponds to the leftmost candidate on the right of the midpoint $(c+c')/2$ of the interval $\C[c, c']$. 
\qed\end{proof}

The numbers of voters $n_l$ and $n_r$ associated with the two subintervals $\C[c,c_l]$ and $\C[c_r, c']$ are correctly computed in steps 24 and 25 using information from the voters' rankings profile $\vec{\succ}$. As for the lengths $d(c, c_l)$ and $d(c_r, c')$ of the two subintervals $\C[c, c_l]$ and $\C[c_r, c']$, if either the steps 14-15 or the steps 20-21 are executed, we need an additional distance query to get them right. If the steps 14-15 are executed, the distance $d(c_r, c')$ is equal to the distance $d(c_l, c')$ in step~13, and we need an additional query for the distance $d(c, c_l)$. If the steps 20-21 are executed, the distance $d(c, c_l)$ is equal to the distance $d(c, c_r)$ in step~19, and we need an additional query for $d(c_r, c')$. 

In all cases, Algorithm~\ref{alg:partitioning} correctly partitions $\C[c, c']$ into $\C[c, c_l]$ and $\C[c_r, c']$, where $c_l$ (resp. $c_r$) is the rightmost (resp. leftmost) candidate on the left (resp. right) of $\C[c, c']$'s midpoint $\mu = (c+c')/2$ and correctly determines $n_l$, $d(c, c_l)$, $n_r$ and $d(c_r, c')$ with at most $4$ distance queries.


\section{Directions for Further Research}
\label{s:conclusions}

Our work opens several interesting directions for further research. 
First, in the proof of Theorem~\ref{thm:coreset}, dependence on $\log n$ seems necessary in order to obtain enough information about the locations of the optimal candidates. Hence, we conjecture a lower bound of $\Omega(k\log n)$ on the number of distance queries required for constant distortion. 

Our results crucially exploit the linear structure of the instance. It would be interesting if bounded distortion can be achieved with a reasonable number of distance queries for the case where the voters and the candidates are embedded in $\reals^d$ and the voters provide a ranking of the candidates in each dimension (since otherwise it is hard to recognize multidimensional Euclidean preferences \cite{Peters17}). 

For general metric spaces, it would be interesting if constant distortion can be achieved with $O(k\log n)$ queries for \emph{perturbation-stable} instances (e.g., \cite{MM21_stability}), where the different clusters of voters are somewhat easier to identify (see also \cite{FotakisP21} for  applications of perturbation stability to mechanism design for $k$-facility location). 


\medskip\medskip\medskip\medskip\medskip\medskip
\appendix
\parbox{\textwidth}{\centering{\LARGE\bf Appendix}}

\medskip
\section{Assuming Non-Degenerate Profiles}
\label{s:app:top-choice}

We next justify why assuming non-degenerate ranking profiles $\vec{\succ}$ is without loss of generality. Using the polynomial-time algorithm of \cite{EF14}, we can deduce from a linear ranking profile $\vec{\succ}$ the linear ordering of the set of candidates $\hat{\C} = \C[c_l, c_r]$ whose location is between the locations of the leftmost active candidate $c_l$ and the rightmost active candidate $c_r$. The next proposition indicates that the set $\hat{\C}$ includes the most interesting candidates to elect, in the sense that electing a candidate that is not in $\hat{\C}$ is not beneficial for the social cost. 

\begin{proposition}\label{prop:tops-only} 
For every $k$-committee $S$ such that $S \not\subseteq \hat{\C}$, there exists a $k$-committee $S' \subseteq \hat{\C}$ with $\Util(S') \le \Util(S)$.  
\end{proposition}

\begin{proof}
Consider a committee $S$ with a candidate $c \not\in \hat{\C}$. We first assume without loss of generality that $c$ is on the left of the leftmost active candidate $c_l$ (the case where $c$ is on the right of the rightmost active candidate $c_r$ is symmetric). No voter prefers $c$ over $c_l$, because by the definition of $c_l$, the top choice of that voter is either $c_l$ or a candidate on the right of $c_l$. If $c_l \in S$, then no voter is assigned to $c$, and $S\setminus \{c\}$ has social cost $\Util(S\setminus \{c\}) = \Util(S)$. We may add any candidate of $\hat{\C} \setminus S$ to $S \setminus \{c\}$, so that we obtain a $k$-committee, without increasing the social cost. If $c_l \not\in S$, then $(S \setminus \{ c \}) \cup \{ c_l \}$ is a $k$-committee with social cost at most $\Util(S)$. We can apply the argument above repeatedly, until we obtain a $k$-committee $S' \subseteq \hat{\C}$ with $\Util(S') \le \Util(S)$. 
\qed\end{proof}

We note that the equivalent of Proposition~\ref{prop:tops-only} for the egalitarian cost is also true and can be proven similarly. 

\section{Low Distortion via Good Subsets of Candidates for the Egalitarian Cost}
\label{s:app:good-egalitarian}


We next show that a $\beta$-approximate $k$-committee $S$ wrt. the egalitarian cost for the candidate-restricted instance $\C_{\cres}$ induced by the original set of candidates $\C$ (i.e., $\C_{\cres}$ is the modified instance where each voter is collocated with her top candidate and all inactive candidates are removed) achieves a distortion of at most $1+2\beta$ wrt. the egalitarian cost for the original instance. 
We note that in the proof of Theorem~\ref{thm:factor3_egal}, the optimal solution $S^\sharp$ for $\C_{\cres}$ wrt. the egalitarian cost ($S^\sharp$ is used as benchmark for the definition of the $\beta$-approximation ratio of $S$) may include inactive candidates from $\C$ (that are not present in $\C_{\cres}$).

\begin{theorem}\label{thm:factor3_egal}
Let $(\C, \V)$ be an instance of the $k$-Committee Election, let $S \subseteq \C$ (resp. $S^\ast \subseteq \C$) be a $\beta$-approximate (resp. an optimal) $k$-committee wrt. the egalitarian cost for the candidate-restricted instance $\C_{\cres}$ (resp. original instance). Then, $\Egal(S) \leq (1+2\beta)\Egal(S^\ast)$.
\end{theorem}

\begin{proof}
We recall that for each voter $v \in \V$, $\top(v)$ is $v$'s top candidate in $\C$. Then, by the triangle inequality, $d(v, S) \leq d(v, \top(v)) + d(\top(v), S)$. Taking the maximum over all voters $v \in \V$, we obtain that:
\begin{align}
  \Egal(S) & \leq \Egal(\C) + \Egal(\C_{\cres}, S)\,,\ \ \text{where} \label{eq:beta_good_egal1} \\
\Egal(\C) & = \max_{v\in \V} \{ d(v, \top(v)) \} = \max_{v\in \V} \{ d(v, \C)\}\ \ \text{and} \notag\\
\Egal(\C_{\cres}, S) & = \max_{v \in \V} \{ d(\top(v), S)\}\,. \notag
\end{align}
$\Egal(\C_{\cres}, S)$ is the egalitarian cost of $S$ for the candidate-restricted instance $\C_{\cres}$ induced by $\C$. We observe that 
\[  \Egal(\C_{\cres}, S) \leq \beta\,\Egal(\C_{\cres}, S^\sharp)
                         \leq \beta\,\Egal(\C_{\cres}, S^\ast)\,,\]
because $S$ is a $\beta$-approximate $k$-committee for $\C_{\cres}$. For the second inequality, we use that the $\beta$-approximation ratio of $S$ in the candidate-restricted instance $\C_{\cres}$ is established against an (unrestricted) optimal solution $S^\sharp$ that may also include inactive candidates from $\C$ not included in $\C_{\cres}$. Hence, $S^\ast$ is also a feasible alternative to $S^\sharp$ as an (unrestricted) optimal solution for $\C_{\cres}$. Therefore, $\Egal(\C_{\cres}, S^\sharp) \leq \Egal(\C_{\cres}, S^\ast)$, because $S^\sharp$ is an optimal solution for $\C_{\cres}$ wrt. the egalitarian cost, and the $\beta$-approximation ratio of $S$ in $\C_{\cres}$ also holds against $S^\ast$. 

Moreover, since $d(\top(v), S^\ast) \leq d(\top(v), v) + d(v, S^\ast)$, by the triangle inequality, we take the maximum over all voters $v$ and obtain that $\Egal(\C_{\cres}, S^\ast) \leq \Egal(\C)+\Egal(S^\ast)$. 

Combined with the observations above, \eqref{eq:beta_good_egal1} implies that:
\[ \Egal(S) \leq (1+\beta)\Egal(\C) + \beta\,\Egal(S^\ast) \leq (1+2\beta)\Egal(S^\ast)\,,\]
where we use that $\Egal(\C_{\cres}, S) \leq \beta\,\Egal(\C_{\cres}, S^\ast)$.
\qed\end{proof}


In the proof of Theorem~\ref{thm:k-queries}, we can use Theorem~\ref{thm:factor3_egal} and obtain a distortion of $5$ for the egalitarian cost (and a distortion of $5n$ for the social cost), because the $2$-approximation ratio of Algorithm~\ref{alg:k-center} (see e.g., \cite[Theorem~2.3]{WS10}) holds against an optimal solution that may also include ``inactive candidates'' (i.e., optimal candidates associated with any demand points). 

For completeness, we next show the equivalent of Theorem~\ref{thm:factor3} for the egalitarian cost (even though we do not use it anywhere). We say that a set $\C' \subseteq \C$ is an $(\ell, \beta)$-good set wrt. the egalitarian cost, if $|\C'| = \ell$ and $\Egal(\C') \leq \beta\,\Egal(S^\ast)$, where $S^\ast \subseteq \C$ is a $k$-committee with optimal egalitarian cost for the original instance. 
The following shows that computing an optimal $k$-committee $S \subseteq \C'$ for the candidate-restricted $\C'_{\cres}$ instance induced by an $(\ell, \beta)$-good set $\C'$ wrt. the egalitarian cost implies a distortion of $2+3\beta$ wrt. the egalitarian cost for the original instance. 

The crucial difference between Theorem~\ref{thm:factor3_egal} and Theorem~\ref{thm:factor5_egal} is that in Theorem~\ref{thm:factor3_egal}, the $\beta$-ap\-prox\-i\-ma\-tion
of $S$ in the candidate-restricted instance $\C_{\cres}$ holds against an optimal $k$-committee $S^\sharp$ that can include both active and inactive candidates from $\C$ (even though inactive candidates are not present in $\C_{\cres}$); while in  Theorem~\ref{thm:factor5_egal}, the optimality of $S$ for the candidate-restricted instance $\C'_{\cres}$ holds against restricted $k$-committees that can include only active candidates present in $\C'_{\cres}$.  

\begin{theorem}\label{thm:factor5_egal}
Let $(\C, \V)$ be an instance of the $k$-Committee Election, let $\C' \subseteq \C$ be an $(\ell, \beta)$-good set of candidates wrt. the egalitarian cost, let $\C'_{\cres}$ be the candidate-restricted instance induced by $\C'$ and let $S$ (resp. $S^\ast$) be an optimal $k$-committee for $\C'_{\cres}$ (resp. for $(\C, \V)$) wrt. the egalitarian cost. Then, $\Egal(S) \leq (2+3\beta)\Egal(S^\ast)$.
\end{theorem}

\begin{proof}
The proof is quite similar to the proof of Theorem~\ref{thm:factor3_egal}. 
For each voter $v \in \V$ (with her location $v$ as in the original instance), we let $\top'(v) \in \C'$ be $v$'s top candidate in $\C'$. By the triangle inequality, $d(v, S) \leq d(v, \top'(v)) + d(\top'(v), S)$. Taking the maximum over all voters $v \in \V$, we obtain that:
\begin{align}
  \Egal(S) & \leq \Egal(\C') + \Egal(\C'_{\cres}, S)\,,\ \ \text{where} \label{eq:beta_good_egal2} \\
\Egal(\C') & = \max_{v\in \V} \{ d(v, \top'(v)) \} 
             = \max_{v\in \V} \{ d(v, \C')\}\ \ \text{and} \notag\\
\Egal(\C'_{\cres}, S) & = \max_{v \in \V} \{ d(\top'(v), S)\}
                        = \max_{i\in[\ell]} \{ d(c_i, S) \}\,, \notag
\end{align}
where $\Egal(\C'_{\cres}, S)$ is the egalitarian cost of $S$ for the candidate-restricted instance $\C'_{\cres}$ induced by $\C'$. 

For each candidate $c_i \in \C'$, let $\top_S(c_i) \in S \subseteq \C'$ be the nearest candidate of $c_i$ in $S$. Then, by the triangle inequality, 
\[ d(c_i, \top_S(c_i)) \leq d(c_i, S^\ast) + d(S^\ast, \top_S(c_i))\,. \]
Taking the maximum over all $c_i \in \C'$, and since $\top_S(c_i) \in \C'$, we obtain that
\begin{equation}\label{eq:beta_good_egal3}
 \Egal(\C'_{\cres}, S) \leq 2\,\Egal(\C'_{\cres}, S^\ast)
\end{equation}

Since $d(\top'(v), S^\ast) \leq d(\top'(v), v) + d(v, S^\ast)$, by the triangle inequality, 
$\Egal(\C'_{\cres}, S^\ast) \leq \Egal(\C')+\Egal(S^\ast)$. 

Combined with the observations above, \eqref{eq:beta_good_egal2} implies that:
\[ \Egal(S) \leq 3\,\Egal(\C') + 2\,\Egal(S^\ast) \leq (2+3\beta)\Egal(S^\ast)\,,\]
where we first use \eqref{eq:beta_good_egal3} and then use the hypothesis that $\C'$ is an $(\ell, \beta)$-good set of candidates wrt. the egalitarian cost.
\qed\end{proof}

\section{Ignoring Inactive Candidates}
\label{s:app:active-candidates}


The following shows that in instances where \emph{every voter is collocated with her top candidate}, we can ignore inactive candidates, i.e., candidates $c$ with $\Cl(c) = \emptyset$ in the rankings profile $\vec{\succ}$, without increasing the social cost of $k$-committees. Hence, when we deal with the distortion wrt. the social cost and compute the optimal $k$-committee of a candidate-restricted instance induced by an $(\ell, \beta)$-good set of candidates, we do not need to consider inactive candidates (see Section~\ref{s:good} for the definition of $(\ell, \beta)$-good sets). Then, when we state and prove the distortion bounds in Theorem~\ref{thm:k-queries} and Theorem~\ref{thm:factor3}, we make sure that they hold against an optimal solution for the original instance, where some candidates may be inactive and candidate and voter locations may be different. 


\begin{proposition}\label{pr:active}
Let $\C$ be the set of all candidates and let $\tilde{\C} \subset \C$ be the set of active candidates. In instances where every voter is collocated with her top candidate, for every committee $S$ that includes inactive candidates, 
%
%
there exists another committee $S' \subseteq \tilde{\C}$ with $\Util(S') \le \Util(S)$.  
\end{proposition}

\begin{proof}
Let $S$ be a $k$-committee that includes an inactive candidate $c \not\in \tilde{\C}$. If no voter is assigned to $c$, we can remove $c$ from $S$. We get a new committee with $k-1$ candidates, less inactive candidates than $S$, and $\Util(S\setminus \{c\}) \leq \Util(S)$. We may add any candidate of $\tilde{\C} \setminus S$ to $S \setminus \{c\}$, so that we obtain a $k$-committee, without increasing the social cost.

Otherwise, let $\V_{\text{left}}$ (resp. $\V_{\text{right}}$) be the set of voters on the left (resp. right) of $c$ that are assigned to $c$ under the $k$-committee $S$. By hypothesis, every voter in $\V_{\text{left}} \cup \V_{\text{right}}$ is collocated with a candidate. 
If $|\V_{\text{left}}|>|\V_{\text{right}}|$, we replace $c$ in $S$ with the candidate $c'$ collocated with the rightmost voter of $\V_{\text{left}}$. Otherwise, we replace $c$ in $S$ with the candidate $c'$ collocated with the leftmost voter of $\V_{\text{right}}$. In both cases, we obtain a new $k$-committee $S' = (S \setminus \{c \}) \cup \{ c'\}$ with less inactive candidates than $S$ and $\Util(S') \leq \Util(S)$.
We can apply the argument above repeatedly, until we obtain a $k$-committee $S' \subseteq \tilde{\C}$ with $\Util(S') \le \Util(S)$. 
\qed\end{proof}

\section{Distortion without Distance Queries}
\label{s:ordinal}

In this section, we consider deterministic rules that do not use distance queries. The results of \cite{CSV22} imply that for $k = 2$, the best possible distortion is $\Theta(n)$, where their lower bound holds for $1$-dimensional instances and their upper bound holds for general metric spaces. The upper bound is achieved by a deterministic rule that when applied to our setting with $1$-Euclidean preferences boils down to selecting the leftmost and the rightmost active candidates. Moreover, \cite{CSV22} proved that the distortion of any rule for $k$-committee election is unbounded for any $k \geq 3$. 

As a warmup and for sake of completeness, we present simple proofs of the above bounds for the special case of $1$-Euclidean preferences. We note that similar results have been obtained by \citet{AnshZ21,CSV22,Pulya22}. Our proofs carefully determine the constants involved. 

\subsection{Linear Lower Bound on the Distortion of $2$-Committee Election}

\begin{theorem}\label{thm:2fac_lower_bound}
The distortion of any deterministic rule for electing $k = 2$ (out of $m \geq 4$) candidates, in a setting where $n \geq 4$ voters have $1$-Euclidean preferences, is at least $n-1$.
\end{theorem}

\begin{proof}
Wlog., we consider an instance with only $4$ candidates $a, b, c, d$, appearing in this order from left to right on the real line. We let:
\[    d(a, b) = x, \ \ d(b, c) = y\ \,\mbox{and}\, \ d(c, d) = z, \]
with $x \geq y+z$ and $y \geq z$. For simplicity, we consider $n=2t+2$ voters, with the following rankings:
\begin{itemize}
\item $t$ voters rank the candidates as $b \succ c \succ d \succ a$\,. 

\item $t$ voters rank the candidates as $c \succ d \succ b \succ a$\,. 
    
\item $1$ voter ranks the candidates as $a \succ b \succ c \succ d$\,. 
    
\item $1$ voter ranks the candidates as $d \succ c \succ b \succ a$\,.
\end{itemize}

We distinguish $6$ different cases, depending on the pair of candidates selected by the deterministic rule for $2$-committee election. Choosing the distances $x, y, z$ appropriately, we end up with a different optimal solution in each case. Taking the ratio of the rule's social cost to the optimal social cost, we lower bound the distortion in each case. The theorem follows by taking the minimum of these lower bounds. 

In the following, we consider all possible candidate pairs that can be elected by the voting rule: 
\begin{description}
\item[$\{a, b\}$ or $\{a, c\}$ or $\{a, d\}$:] Let $x=1+\epsilon$, $y=1$ and $z=\epsilon$, for some small $\epsilon > 0$. Thus, the candidates $c$ and $d$ are almost collocated. Suppose that the voter with profile $a \succ b \succ c \succ d$ is close to the middle of $[a,b]$, at $1/2$, whereas all other voters are collocated with their top choice. The optimal solution is to elect $\{b, c\}$. If $\{a, b\}$ is elected by the voting rule, then the distortion becomes $2t+3$ as $\epsilon \rightarrow 0$. If $\{a, c\}$ or $\{a, d\}$ is elected by the voting rule, then the distortion becomes $2t+1$ as $\epsilon \rightarrow 0$.  

\item[$\{b, c\}$:] Suppose that every voter is collocated with her top candidate. The rule's social cost is $x+z$. The optimal solution is to elect $\{a, c\}$ and has social cost $ty+z$. Then, $x$ can become arbitrarily large (compared against $t$, $y$ and $z$), which implies that the distortion cannot be bounded by any function of $t$ (or $n$).

\item[$\{b, d\}$ or $\{c, d\}$:] Suppose that every voter is collocated with her top candidate. The rule's social cost is at least $x+ty$. The optimal solution is to elect $\{a, c\}$ and has social cost $ty+z$. As before, $x$ can become arbitrarily large (compared against $t$, $y$ and $z$), which implies that the distortion that cannot be bounded by any function of $t$ (or $n$).
\end{description}

Hence, we get a distortion of at least $2t+1=n-1$. \qed\end{proof}

An interesting corollary of the proof of Theorem~\ref{thm:2fac_lower_bound} is that the distortion can be unbounded, unless we select the top candidate of the most distant extreme voter. Since in general we cannot identify the most distant extreme voter without resorting to distance queries, a safe choice is to elect the leftmost and the rightmost active candidates. 

\subsection{The Distortion of Electing the Two Extreme Candidates}

\begin{theorem}\label{thm:2fac_upper_bound}
For any number $n \geq 3$ of voters with $1$-Euclidean preferences, the $2$-committee rule that elects the leftmost and the rightmost active candidates achieves a distortion of at most $2n-2$. 
\end{theorem}

\begin{proof}
We next show that electing the leftmost active candidate and the rightmost active candidate achieves a distortion at most $2n-2$, where $n \geq 3$ is the number of voters (we note that for $n = 2$ voters, we can elect their top candidates, getting an optimal committee). 

Let $\{c, c'\}$, with $c < c'$, be the optimal committee, and let $C_1$ (resp., $C_2$) be the set of voters preferring $c$ to $c'$ (resp., $c'$ to $c$). We analyze each optimal cluster of voters, $C_1$ and $C_2$, separately. 

We let $v$ (resp. $v'$) be the leftmost (resp. rightmost) voter in $C_1$ (resp. $C_2$), and let $c_l$ be $v$'s (resp. $c_r$ be $v'$'s) top candidate, which is the leftmost (resp. rightmost) candidate appearing as a top candidate of some voter in $C_1$ (resp. $C_2$). Hence, $c_l$ is the leftmost active candidate and $c_r$ is the rightmost active candidate, and thus, $\{ c_l, c_r \}$ is the committee elected by our voting rule. 

Let $|C_1| = n_1$, let $\OPT_1 = \sum_{x \in C_1} d(c, x)$ be the social cost of the voter cluster $C_1$ in the optimal committee. We let $\A_1 = \sum_{x \in C_1} d(c_l, x)$ be an upper bound on the social cost of the voter cluster $C_1$ for committee $\{ c_l, c_r \}$. Then, observing that $d(v, c_l) \leq d(v, c)$, we get that:
\begin{align*}
\A_1 & = d(v, c_l) + \sum_{x \in C_1 \setminus \{ v \}} d(x, c_l) \\
     & \leq d(v, c_l) + \sum_{x \in C_1 \setminus \{ v \}} \big( d(x, c) + d(c, c_l) \big) \\
     & \leq \OPT_1 + (n_1-1) d(c, c_l) \\
     & \leq \OPT_1 + (n_1-1) \big(d(c, v) + d(v, c_l)\big) \\
     & \leq \OPT_1 + 2(n_1-1) d(c, v) \\ 
     & \leq (2n_1-1) \OPT_1 \,.
\end{align*}
The first and the third inequalities follow from the triangle inequality. The fourth inequality holds because $d(v, c_l) \leq d(v, c)$. The last inequality follows from the fact that $\OPT_1 \geq d(v, c)$. The second inequality holds because 
\[ \OPT_1 = \sum_{x \in C_1} d(x, c) \geq d(v, c_l) + \sum_{x \in C_1 \setminus \{ v \}} d(x, c)\,,\]
since $d(v, c_l) \leq d(v, c)$. 

Similarly, we let $|C_2| = n_2$ and let $\OPT_2 = \sum_{x \in C_2} d(c', x)$ denote the social cost of the voter cluster $C_2$ in the optimal committee. We let $\A_2 = \sum_{x \in C_2} d(c_r, x)$ be an upper bound on the social cost of the voter cluster $C_2$ for committee $\{ c_l, c_r \}$. Applying the same analysis for the voters in $C_2$, we get that $\A_2 \leq (2n_2-1)\OPT_2$\,.

Summing everything up, we get that the social cost of our voting rule is at most:
\begin{align*}
\A_1 + \A_2 
   & \leq (2n_1-1)\OPT_1+(2n_2-1)\OPT_2 \\
   & \leq \big(2(n_1+n_2)-2\big)\OPT \\
   & = 2(n-1)\OPT\,.
\end{align*} 
The second inequality above follows from $\OPT = \OPT_1+\OPT_2 \geq \max\{\OPT_1, \OPT_2\}$. For the equality, we use that $n = n_1+n_2$.
\qed\end{proof}

The analysis of Theorem~\ref{thm:2fac_upper_bound} is practically tight. This is shown by an instance with $3$ candidates $a, b, c$ located at $0$, $2$ and $x$, respectively, where $x \gg 2$, and $n$ voters, one voter located at $1-\epsilon$, $n-2$ voters located at $2$ and one voter located at $x$. The optimal committee is to elect $\{b, c\}$ and has social cost $1+\epsilon$. Our voting rule elects $\{a, c\}$ and has social cost $2(n-2)+(1-\epsilon) = 2n-3-\epsilon$. Hence, we get a distortion of $2n-3$, as $\epsilon \rightarrow 0$. In fact, for $\OPT_2 = 0$ and $n_2 = 1$, which hold for this instance, the analysis of Theorem~\ref{thm:2fac_upper_bound} gives an upper bound of $2n-3$ on the distortion.

\subsection{The Distortion of Voter Clustering Based on the Two Extreme Candidates}

We next show that clustering the voters based on the two extreme candidates and selecting the top candidate of the median voter in each clusters achieves a distortion of at most $n+1$. Specifically, we let $c_l$ (resp. $c_r$) denote the leftmost (resp. rightmost) active candidate in $\C$, and let $A_1 = \{ v \in \V : c_l \succ_v c_r \}$ and $A_2 = \{ v \in \V : c_r \succ_v c_l \}$ be the sets of voters preferring $c_l$ to $c_r$ and $c_r$ to $c_l$, respectively. Namely, $A_1$ and $A_2$ are the voter clusters induced by the two extreme active candidates. We also let $v^m_1$ and $v^m_2$ denote the median voters of $A_1$ and $A_2$ respectively. Then, we elect the top candidates $a_1 = \top(v_1^m)$ and $a_2 = \top(v_2^m)$ of the two median voters $v_1^m$ and $v_2^m$ ($a_1$ and $a_2$ can be computed using the candidate ordering on their axis, provided by the algorithm of \cite{EF14}, and the cardinality of the candidate clusters).

\begin{theorem}\label{thm:2fac_upper_bound_improved}
For any number $n \geq 3$ of voters with $1$-Euclidean preferences, electing the top candidates $a_1 = \top(v_1^m)$ and $a_2 = \top(v_2^m)$ of the two median voters $v_1^m$ and $v_2^m$ in the clusters $A_1$ and $A_2$, induced by the two extreme active candidates, achieves a distortion of at most $n+1$. 
\end{theorem}

\begin{proof}
Let $\{c_1, c_2\}$, with $c_1 < c_2$, be the optimal committee, and let $C_1$ (resp., $C_2$) be the set of voters preferring $c_1$ to $c_2$ (resp., $c_2$ to $c_1$); i.e., $C_1$ and $C_2$ are the optimal voter clusters. Wlog., we assume that $A_1 \subset C_1$ and $C_2 \subset A_2$ (if $A_1 = C_1$ and $A_2 = C_2$, the distortion is at most $3$, as shown e.g., in \cite[Section~3.3]{FG22}, while the case where $C_1 \subset A_1$ and $A_2 \subset C_2$ is symmetric). We let  $d^\ast_v = d(v, \{c_1, c_2\})$ denote each voter's $v$ contribution to the optimal social cost, let $\OPT = \sum_{v \in \V} d^\ast_v$ be the optimal social cost, and let $n_1 = |A_1|$ and $n_2 = |A_2|$ be the number of voters in each algorithm's cluster. 
%
%
For locations $x, y \in \reals$, we let $(x-y)_+ = \max\{x-y, 0\}$.

In our setting with $1$-Euclidean preferences, if we consider single-winner elections (as in the case where we elect a single candidate to serve the voters in $A_1$ based on their ordinal preferences), selecting the top candidate of the median voter has a distortion of at most $3$ (see also e.g. the discussion in \cite[Section~3.3]{FG22}). Therefore, using that the optimal social cost of voters in $A_1 \subset C_1$ cannot be larger than their social cost wrt. $c_1$, we obtain that:
\begin{equation}\label{eq:alg1_cost}
 \A_1 := \sum_{v \in A_1} d(v, a_1) \leq 3 \sum_{v \in A_1} d^\ast_v \,,
\end{equation}

To analyze the distortion of $a_2$ for the voters in $A_2$ compared against their social cost in the optimal 2-committee $\{c_1, c_2\}$, we partition the voters of $A_2$ into the following sets: 
\begin{align*}
 L & = \{ v \in A_2 : v < v_2^m \mbox{ and } d^\ast_v < d(v, a_2) \} \\
 M & = \{ v \in A_2 : d(v, a_2) \leq d^\ast_v \} \\
 R & = \{ v \in A_2 : v > v_2^m \mbox{ and } d^\ast_v < d(v, a_2) \}
\end{align*}
Namely, $L$ (resp. $R$) is the set of voters in $A_2$ on the left (resp. on the right) of $v_2^m$ that are closer to the optimal committee $\{ c_1, c_2 \}$ than to $a_2$. We note that for every voter $v$ collocated with $v_2^m$, $d(v, a_2) \leq d^\ast_v$, because $a_2$ is the top candidate of $v_2^m$ (and thus, of $v$). Hence, 
the sets $L$, $M$ and $R$ are mutually disjoint and have $L \cup M \cup R = A_2$. We also note that $|L|, |R| \leq n_2/2$, because $v_2^m$ is the median voter of $A_2$ and the voters in $L$ (resp. $R$) lie strictly on the left (resp. on the right) of $v^m_2$, and that $L$ or $R$ may be empty. 


In the following, we let $v^l_2 \leq v^m_2$ be the leftmost voter in $A_2$. We next show that:

\begin{equation}\label{eq:alg2}
 \A_2 := \sum_{v \in A_2} d(v, a_2) \leq 
 \left\{\begin{array}{ll}
   \sum_{v \in A_2} d^\ast_v 
 + \frac{n_2}{2} d(v^l_2, c_2) & \mbox{if $v^l_2 \leq a_2 \leq c_2$} \\
   \sum_{v \in A_2} d^\ast_v 
 + \frac{n_2}{2} \big( d(a_2, v^l_2) + d(v^l_2, c_2) \big) \ \ \ \ & 
    \mbox{if $a_2 \leq c_2$ and $a_2 < v^l_2$} \\
 \sum_{v \in A_2} d^\ast_v 
 + \frac{n_2}{2} (a_2 - v^l_2)_+ & 
    \mbox{if $c_2 < a_2$}
\end{array}\right.
\end{equation}

\begin{proof}[of \eqref{eq:alg2}] 
To upper bound the contribution of voters in $M$ to $\A_2$, we use that:
\begin{equation}\label{eq:cost_middle}
 \sum_{v \in M} d(v, a_2) \leq \sum_{v \in M} d^\ast_v 
\end{equation}

To upper bound the contribution of voters in $L$ to $\A_2$, we observe that every $v \in L$ lies on the left of $a_2$, i.e., every $v \in L$ has $v < a_2$. To justify this claim, we observe that if $v^m_2 > a_2$, since $a_2$ is the top candidate of $v^m_2$, i.e. since $d(v^m_2, a_2) < d(v^m_2, \C \setminus \{a_2\})$, there is no candidate other than $a_2$ in $[a_2, v^m_2 + d(v^m_2, a_2)]$. Therefore, every voter $v$ with $a_2 \leq v < v^m_2$ must have $a_2$ as their top candidate and must be included in $M$. Therefore, either $a_2 \leq v^l_2$, in which case $L = \emptyset$, or $v^l_2 < a_2$, in which case every $v \in L$ has $v < a_2$ and $d(v, a_2) \leq d(v^l_2, a_2) = (a_2 - v^l_2)_+$. Since $|L| \leq n_2/2$, we obtain that:
\begin{equation}\label{eq:cost_left}
 \sum_{v \in L} d(v, a_2) \leq \frac{n_2}{2}\,(a_2 - v^l_2)_+ 
\end{equation}

To upper bound the contribution of voters in $R$ to $\A_2$, we first observe that every $v \in R$ lies on the right of $a_2$, i.e., $v$ has $a_2 < v$, which follows from the symmetric argument of that about the voters in $L$. Then, we distinguish between the case where $a_2 \leq c_2$ and the case where $c_2 < a_2$. 

We start with the case where $a_2 \leq c_2$. In this case, every $v \in R$ has $d(v, a_2) \leq d^\ast_v + d(a_2, c_2)$. More specifically, since $v \in R$, $v > a_2$. Then, either $a_2 < v \leq c_2$, in which case $d(v, a_2) \leq d(a_2, c_2)$, or $a_2 \leq c_2 < v$, in which case $d(v, a_2) \leq d(v, c_2) + d(a_2, c_2) = d^\ast_v + d(a_2, c_2)$. 

We proceed by further distinguishing two subcases depending on whether $v^l_2 \leq a_2$ or $a_2 < v^l_2$. We start with the case where $v^l_2 \leq a_2 \leq c_2$. Using that for every $v \in R$, $d(v, a_2) \leq d^\ast_v + d(a_2, c_2)$, and that $|R| \leq n_2/2$, we obtain the following upper bound about the voters in $R$: 
\begin{equation}\label{eq:cost_right} 
 \sum_{v \in R} d(v, a_2) \leq \sum_{v \in R} d^\ast_v + \frac{n_2}{2} d(a_2, c_2) 
\end{equation}
Using that $d(v^l_2, a_2) + d(a_2, c_2) = d(v^l_2, c_2)$, because of our assumption that $v^l_2 \leq a_2 \leq c_2$, and combining \eqref{eq:cost_middle}, \eqref{eq:cost_left} and \eqref{eq:cost_right}, we obtain the first case of \eqref{eq:alg2}.

We next consider the case where $a_2 \leq c_2$ and $a_2 < v^l_2$. The latter implies that $L = \emptyset$ in this case. Moreover, we observe that in this case, for every $v \in R$, 
\begin{align*}
d(v, a_2) \leq d^\ast_v + d(a_2, c_2) 
          \leq d^\ast_v + d(a_2, v^l_2) + d(v^l_2, c_2)
\end{align*} 
%
Using that $|R| \leq n_2/2$, and combining the inequality above with \eqref{eq:cost_middle} and the fact that $L = \emptyset$, we obtain the second case of \eqref{eq:alg2}.  

Finally, we consider the case where $c_2 < a_2$, which since $d(v^m_2, a_2) < d(v^m_2, c_2)$, implies that $c_2 < v^m_2$. In this case, $R = \emptyset$, because for every $v \in A_2$ with $v > v^m_2$, we have that $d(v, a_2) < d(v, c_2) = d^\ast_v$ (and hence $v$ is not included in $R$). To justify the last inequality, we observe that if $c_2 < a_2 \leq v^m_2$, then $a_2$ lies between $v$ and $c_2$, while if $c_2 < v^m_2 < a_2$, then
\[ d(v, a_2) < d(v^m_2, a_2) < d(v^m_2, c_2) < d(v, c_2) \,. \] 
%
%
Then, combining \eqref{eq:cost_middle} with \eqref{eq:cost_left} and with the fact that $R = \emptyset$, we obtain the third case of \eqref{eq:alg2}.
\qed\end{proof}

We next prove a sequence of useful lower bounds on $\OPT$. We first show that $\OPT \geq d(v^l_2, c_l)/2$. To this end, we let $v^l_1$ be the leftmost voter in $A_1$. Then, $c_l = \top(v^l_1)$ and $\OPT \geq d(c_l, v^l_1)$. Moreover, $\OPT \geq d(v^l_1, v^l_2) \geq d(v^l_2, c_l) - d(c_l, v^l_1)$, because in the optimal solution, $v^l_1, v^l_2 \in C_1$ and are both assigned to $c_1$. Therefore, we conclude that $\OPT \geq d(v^l_2, c_l)/2$. Then, we observe that $d(v^l_2, c_r) < d(v^l_2, c_l)$, because $v^l_2$ prefers $c_r$ to $c_l$ and is included in $A_2$. Observing that $c_2 \leq c_r$ and putting everything together, we obtain the following sequence of lower bounds on $\OPT$:
\begin{equation}\label{eq:opt_lower_bound}
\OPT \geq \frac{d(v^l_2, c_l)}{2} > \frac{d(v^l_2, c_r)}{2} \geq \frac{d(v^l_2, c_2)}{2}\,. 
\end{equation}

Using \eqref{eq:opt_lower_bound}, we next show that in all cases of \eqref{eq:alg2}, 
\begin{equation}\label{eq:cost_final}
 \A_2 \leq 2 \sum_{v \in A_2} d^\ast_v + n_2\,\OPT  
\end{equation}

In the first case of \eqref{eq:alg2}, we use that $\OPT \geq d(v^l_2, c_2) / 2$, from \eqref{eq:opt_lower_bound}, and obtain \eqref{eq:cost_final}.

In the third case of \eqref{eq:alg2}, we observe that $(a_2 - v^l_2)_+ \leq d(v^l_2, c_r)$, because (i) $c_r$ is the leftmost active candidate, and thus $a_2 \leq c_r$; and (ii) $a_2 - v^l_2$ is positive only if $v^l_2 < a_2$, in which case $a_2 - v^l_2 \leq d(v^l_2, c_r)$. Therefore, we use that $\OPT \geq d(v^l_2, c_r) / 2$, from \eqref{eq:opt_lower_bound}, and obtain \eqref{eq:cost_final}.

In the second case of \eqref{eq:alg2}, we use $\OPT \geq d(v^l_2, c_2) / 2$, from \eqref{eq:opt_lower_bound}, and obtain that $n_2 d(v_2^l, c_2) / 2 \leq n_2\,\OPT$. To upper bound $n_2 d(a_2, v_2^l) / 2$, we observe that:
\begin{equation}\label{eq:cost_third_term}
  \frac{n_2}{2} d(v^l_2, a_2) \leq 
  \sum_{v \in M} d(v, a_2) \leq
  \sum_{v \in M} d^\ast_v
\end{equation} 
The first inequality in \eqref{eq:cost_third_term} holds because (i) $|M| \geq n_2/2$, since $|R| \leq n_2/2$ and in this case $L = \emptyset$; and (ii) $a_2 < v^l_2 \leq v$, for every $v \in M$. The last inequality is \eqref{eq:cost_middle}. Applying the upper bounds above to the second case of \eqref{eq:alg2}, we obtain 
\eqref{eq:cost_final}.

If $n_1 = 1$ and $n_2 = n-1$, \eqref{eq:alg1_cost} is simply $\A_1 = d(v_1, a_1) \leq d^\ast_{v_1}$, where $v_1$ denotes the single voter in $A_1$. Then, using \eqref{eq:cost_final}, we conclude that
\[ \A_1 + \A_2 \leq 2\,\OPT + (n-1)\OPT = (n+1)\OPT \]
If $n_1 \geq 2$ and $n_2 \leq n-2$, we combine \eqref{eq:alg1_cost} and \eqref{eq:cost_final} and obtain that:
\[ \A_1 + \A_2 \leq 3\,\OPT + (n-2)\OPT = (n+1)\,\OPT \]
Hence, in all cases, the distortion is at most $n+1$, which concludes the proof of Theorem~\ref{thm:2fac_upper_bound_improved}.
\qed\end{proof}

The analysis of Theorem~\ref{thm:2fac_upper_bound_improved} is tight. This is shown by an instance with $m=3$ candidates $a, b$ and $c$ located at $0$, $1+\epsilon$ and $2+\epsilon$, respectively, for an appropriately small $\epsilon > 0$, and $n \geq 4$ voters, where one voter is on $1/2$, a second voter is on $1$, $n/2-1$ voters are on $1+\epsilon$ and $n/2-1$ voters are on $2+\epsilon$. Then, cluster $A_1$ consists of the two voters on $1/2$ and $1$ and elects $a_1 = a$ (assuming that ties for the median voter are broken in favor of the leftmost median voter for the left cluster). Cluster $A_2$ consists of the $n/2-1$ voters on $1+\eps$ and the $n/2-1$ voters on $2+\eps$ and elects $a_2 = c$ (assuming that ties for the median voter are broken in favor of the rightmost median voter for the right cluster). The social cost of committee $\{ a, c \}$ is $1/2+1+n/2-1 = (n+1)/2$. The optimal $2$-committee is to elect $\{b, c\}$. The optimal social cost is $1/2+2\eps$. Thus we get a distortion of $n+1$, as $\epsilon \rightarrow 0$. If we insist on consistent tie-breaking in favor of the leftmost median voter in both clusters, a variant of this example, where we have $n/2-2$ voters on $1+\eps$ and $n/2$ voters on $2+\eps$, gives a lower bound of $n-1$ on the distortion. Hence, even though the analysis of Theorem~\ref{thm:2fac_upper_bound_improved} is tight (in the sense that $n+1$ is the best possible distortion of this rule, if we do not assume anything about tie-breaking), it is open if consistent tie-breaking among the two clusters could lead to an improved distortion that matches the lower bound of Theorem~\ref{thm:2fac_lower_bound}. 

\subsection{The Distortion of Electing $2$ out of $3$ Candidates}

\begin{proposition}\label{pr:2-out-of-3}
For any number of voters $n \geq 3$ with $1$-Euclidean preferences, the distortion of electing $k = 2$ out of $m = 3$ candidates on the real line is $3$. 
\end{proposition}

\begin{proof}
For the proof, we consider $3$ active candidates $a$, $b$, $c$, which are arranged as $a < b < c$ on the line from left to right. 

\medskip\noindent\textbf{Lower Bound.} We consider $3$ voters with preferences $a \succ b \succ c$, $b \succ c \succ a$, and $c \succ b \succ a$, respectively. 
If we do not elect $a$, the distortion can be unbounded when $d(a,b) \gg d(b,c)$ and each voter is collocated with her top candidate. 
If we do not elect $b$, the distortion can be arbitrary close to $3$, when $d(a,b)=d(b,c)=2$, voter 1 is on $a$, voter 2 is on $b$, and voter 3 is on $(b+c)/2+\eps$, for some small $\eps > 0$. The social cost of $\{a, c\}$ is $3-\eps$, while the social cost of the optimal $2$-committee $\{ a, b \}$ is $1+\eps$. 
If we do not elect $c$, the distortion can be arbitrary close to $3$, when $d(a,b)=d(b,c)=2$, voter 1 is on $a$, voter 2 is on $(b+c)/2-\eps$, for some small $\eps > 0$, and voter 3 is on $c$. The social cost of $\{a, b\}$ is $3-\eps$, while the social cost of the optimal $2$-committee $\{ a, c \}$ is $1+\eps$. 
Hence for any deterministic $2$-committee election rule, there is a $3$-voter ranking profile for which the distortion is arbitrarily close to $3$. 

\medskip\noindent\textbf{Upper Bound.} We assume that all $3$ candidates are active (otherwise, it is optimal to elect the two active candidates). We observe that there are only $4$ possible rankings $a \succ b \succ c$, $b \succ a \succ c$, $b \succ c \succ a$, and $c \succ b \succ a$, which are consistent with the candidate axis $a < b < c$. 
A deterministic voting rule with distortion $3$ is the following:
\begin{itemize}
\item If the number of voters with ranking $a \succ b \succ c$ is at least the number of voters with ranking $b \succ a \succ c$, we elect $a$; otherwise, we elect $b$. 

\item If the number of voters with ranking $b \succ c \succ a$ is at least the number of voters with profile $c \succ b \succ a$, we elect $b$, otherwise, we elect $c$. 

\item If $b$ is elected in both steps, also elect $a$.  
\end{itemize}
Namely, for each pair of consecutive candidates on the axis, we elect the candidate with largest support, if we restrict our attention to the subset of voters with those two candidates as their top two preferences. 

To show that an upper bound of $3$ on the distortion, we partition the voters into those with $c$ as their last candidate and those with $a$ as their last candidate. For both groups, we use the fact that the distortion is at most $3$ (see e.g., \cite{AnshelevichBEPS18}), when we elect $1$ out of $2$ candidates (elect one candidate out of $a$ and $b$ for the former group of voters and elect one candidate out of $b$ and $c$ for the latter group). Therefore, the distortion for the entire set of voters is at most $3$. 
\qed\end{proof}

\subsection{Unbounded Distortion of Election $k \geq 3$ Candidates}

\begin{theorem}\label{thm:3fac_lower_bound}
The distortion of any deterministic rule for electing $k \geq 3$ out of $m \geq 4$ candidates arranged on the real line cannot be bounded by any function of $n$ and $m$.  
\end{theorem}

\begin{proof}
The proof is a simplified version of our more general lower bound in Theorem~\ref{thm:query-lower-bound}. For the proof, we consider $m = 4$ candidates $a$, $b$, $c$ and $d$, which are arranged as $a < b < c < d$ on the line from left to right, and $n = 4$ voters, each collocated with one of the candidates. 

The first voter is collocated with candidate $a$ and has ranking $a \succ_a b \succ_a c \succ_a d$, the second voter is collocated with $b$ and has ranking $b \succ_b a \succ_b c \succ_b d$, the third voter is collocated with $c$ and has ranking $c \succ_c d \succ_c b \succ_c a$, and the fourth voter is collocated with $d$ and has ranking $d \succ_d c \succ_d b \succ_d a$. 

Since we elect $k = 3$ candidates out of $m = 4$, we just need to identify which one is dropped. If either $a$ or $b$ is dropped (the case where either $c$ or $d$ is dropped is symmetric), then we let $d(a,b)=d(b,c)=B \gg 1$ and $d(c,d)=1$. The optimal choice is to drop $d$, resulting in a distortion of $B$, which can become arbitrarily large. The number of active candidates can be increased to any $m \geq 4$ by introducing $m-4$ additional candidates (and $m-4$ additional voters, each collocated with a different new candidate). We divide the new candidates among the $4$ original ones, each placed in a different location essentially collocated with $a$, $b$, $c$ or $d$. Clearly, any of them can be selected in place of the corresponding original candidate, so this modification does not affect the lower bound construction. 
\qed\end{proof}

We also observe that the proof of Theorem~\ref{thm:3fac_lower_bound} 
can be extended to the case where we want to elect $k = m-1$ candidates out of $m$, for any $m \geq 4$.

\section{Simulation among Different Query Types}
\label{s:app:queries}

We first show how to obtain the answer to a voter distance query by combining the answers of two regular queries.

\begin{proposition}[Voter Query Simulation]\label{pr:voter_queries}
The distance between any two voters $v$ and $v'$ can be obtained from two regular cardinal queries, one for $v$ and one for $v'$. 
\end{proposition}	

\begin{proof}
Let $y = d(v, v')$ be the distance of $v$ to $v'$. We distinguish three cases depending on the candidate clusters to which $v$ and $v'$ belong. 

\begin{enumerate}
\item If $v$ and $v'$ have the same top candidate $c$ (i.e., $v, v' \in \Cl(c)$), let $c'$ be any candidate different from $c$ (i.e., $c'$ may be the first candidate to the right of $c$ or the first candidate to the left of $c$). Then, $y = | d(v, c') - d(v', c')|$. 

\item If $v$'s top candidate is $c$ and $v'$'s top candidate is $c' \neq c$, and $c$ and $c'$ are consecutive on the line, let $c''$ be any candidate different from both $c$ and $c'$ (i.e., $c''$ may be the first candidate to the right of the rightmost of $c$ and $c'$ or the first candidate to the left of the leftmost of $c$ and $c'$). Then, $y = | d(v, c'') - d(v', c'')|$. 

\item If $v$'s top candidate is $c$ and $v'$'s top candidate is $c' \neq c$, and $c$ and $c'$ are not consecutive on the line, let $c'' \not\in \{c, c'\}$ be any candidate lying between $c$ and $c'$. Then, $y = d(v, c'') + d(v', c'')$. 
\end{enumerate}

In all cases, we can obtain the value of $y$ from the responses to two regular queries, one about the distance of $v$ to a carefully chosen candidate and another about the distance of $v'$ to the same candidate. 
\qed\end{proof}

We next show how to obtain the answer to a candidate  query by combining the answers of at most six regular queries.

\begin{proposition}[Candidate Query Simulation]\label{pr:candidate_queries}
The distance $d(c, c')$ of any two candidates $c$ and $c'$ can be obtained from 
at most six regular queries, at most three for a voter in $\Cl(c)$ and at most three for a voter in $\Cl(c')$.   
\end{proposition}	

\begin{proof}
We let $v$ be any voter in $\Cl(c)$ and $v'$ be any voter in $\Cl(c')$ (in fact $v$ and $v'$ can be any voters). For brevity, we let $a = d(v, c)$, $b = d(v, c')$, $x = d(v', c')$ and $y = d(v', c)$ throughout this proof. We let $z = d(c, c')$ be the information we aim to extract from $a$, $b$, $x$ and $y$. By symmetry, we can assume that $c < c'$. Then, due to fact that the agents are located on the real line, we have that $\min\{ v, c \} \leq \max\{ v, c \} < \min\{ v', c' \} \leq \max\{ v', c'\}$. We distinguish four cases about how $z$ can be obtained from $a$, $b$, $c$ and $d$, depending on the relative order of $v$, $c$, $v'$ and $c$ on the real line. 

\begin{enumerate}
\item $v \leq c < c' \leq v'$. Then, $z = b - a = y - x$. 
\item $v \leq c < v' \leq c'$. Then, $z = b - a = y + x$. 
\item $c \leq v < c' \leq v'$. Then, $z = b + a = y - x$. 
\item $c \leq v < v' \leq c'$. Then, $z = b + a = y + x$. 
\end{enumerate}

Since we do not have ties, $b > 0$ and $y > 0$. We first consider the case where both $a > 0$ and $x > 0$. Then, if $b \neq y$, the four cases above are mutually exclusive. More specifically, if $b > y$, then:
\begin{enumerate}
\item either $b = a+x+y$, we are in case (2) and output $z = b-a$,
\item or $b \neq a+x+y$ and $a > x$, we are in case (1) and output $z = b-a$,
\item or $b \neq a+x+y$ and $a < x$, we are in case (4) and output $z = b+a$.
\end{enumerate}
Symmetrically, if $b < y$, then:
\begin{enumerate}
\item either $y = a+b+x$, we are in case (3) and output $z = b+a$, 
\item or $y \neq a+b+x$ and $a > x$, we are in case (4) and output $z = b+a$,
\item or $y \neq a+b+x$ and $a < x$, we are in case (1) and output $z = b-a$.
\end{enumerate}
If $b = y$, we can be neither in case (2) (because then $b = a+x+y > y$) nor in case (3) (because then $y = a+b+x > b$). Hence, if $b = y$, we can be either in case (1) or in case (4). In both cases, we also have $a = x$, which implies that we cannot distinguish between them. To resolve the ambiguity, we query $w = d(v, v')$, which can be extracted from one regular query for $v$ and one regular query for $v'$, as described in the proof of Proposition~\ref{pr:voter_queries}. Then, 
\begin{enumerate}
\item either $w = y+a$, we are in case (1) and output $z = b-a$,
\item or $w = y-a$, we are in case (4) and output $z = b+a$.
\end{enumerate}

If both $a = 0$ and $x = 0$, we output $z = b = y$. If $a = 0$ and $x > 0$, we output $z = b$ (in this case, either $c' < v'$ and $z = b = y - x$, or $v' < c'$ and $z = b = y+x$). If $a > 0$ and $x = 0$, we output $z = y$ (in this case, either $c < v$ and $z = y = b - a$, or $v < c$ and $z = y = b+a$). 

In all cases, we obtain the value of $z$ from the responses of at most three regular queries for a voter $v \in \Cl(c)$ and at most three regular queries for a voter $v' \in \Cl(c')$. 
\qed\end{proof}


\bibliography{references_merged}

\end{document}